\documentclass[11pt,letterpaper]{article}

\usepackage[utf8]{inputenc}
\usepackage[margin=1in]{geometry}
\usepackage{graphics} 
\usepackage{tikz}
\usetikzlibrary{positioning}
\usetikzlibrary{arrows,automata}
\usepackage{booktabs, array}
\usepackage{pgf}
\usepackage{natbib} 
\usepackage{epstopdf}
\usepackage{epsfig} 
\usepackage{subcaption}
\usepackage{psfrag}
\usepackage[width=.85\textwidth]{caption}
\usepackage{lipsum}
\usepackage{amsmath} 
\usepackage{amssymb}  
\usepackage{mathbbol}
\usepackage{cite}
\usepackage{wrapfig}
\usepackage{mathrsfs}
\usepackage{url}
\usepackage{amsmath,amsfonts, amssymb,color}
\usepackage{amsthm}
\usepackage[width=.85\textwidth, font=footnotesize]{caption}
\usepackage[pdfpagelabels,pdfpagemode=None,breaklinks=true]{hyperref} \hypersetup{
     colorlinks   = true,
     citecolor    = blue
}
\usetikzlibrary{shapes,arrows,calc}
\usepackage{marvosym}

\renewcommand{\cite}{\citep}

\newtheorem{theorem}{Theorem}

\newtheorem{remark}{Remark}

\newtheorem{assumption}{Assumption}

\makeatletter
\renewcommand\bibsection%
{
  \section*{\refname
    \@mkboth{\MakeUppercase{\refname}}{\MakeUppercase{\refname}}}
}
\makeatother

\newcommand{\ignore}[1]{}

\newcommand*{\TitleFont}{%
      \usefont{\encodingdefault}{\rmdefault}{}{n}%
      \fontsize{16}{20}%
      \selectfont}

\DeclareMathAlphabet{\mathcal}{OMS}{cmsy}{m}{n}      

\newcommand{\Pb}{\mathbb{P}}
\newcommand{\Rb}{\mathbb{R}}

\newcommand{\0}{\mathbf{0}_n}
\newcommand{\1}{\mathbf{1}_n}

\newcommand{\diag}{\mathtt{diag}}

\newcommand{\oprocendsymbol}{\hbox{$\bullet$}}
\newcommand{\oprocend}{\relax\ifmmode\else\unskip\hfill\fi\oprocendsymbol}

\allowdisplaybreaks

\newcommand{\rev}[1]{{#1}}

\begin{document}

\title{\TitleFont A Closed-Loop Framework for Inference, Prediction and Control of SIR Epidemics on Networks}
\author{~Ashish~R.~Hota, Jaydeep~Godbole and Philip~E.~Par\'e\thanks{Ashish R. Hota and Jaydeep Godbole are with the Department of Electrical Engineering, IIT Kharagpur, India. Philip E. Par\'e is with the School of Electrical and Computer Engineering, Purdue University, USA. E-mail: ahota@ee.iitkgp.ac.in, jgodbole7@gmail.com, philpare@purdue.edu.}}%

\date{}
\maketitle

\begin{abstract}
Motivated by the ongoing pandemic COVID-19, we propose a closed-loop framework that combines inference from testing data, learning the parameters of the dynamics and optimal resource allocation for controlling the spread of the susceptible-infected-recovered (SIR) epidemic on networks. Our framework incorporates several key factors present in testing data, such as the fact that high risk individuals are more likely to undergo testing. We then present two tractable optimization problems to evaluate the trade-off between controlling the growth-rate of the epidemic and the cost of non-pharmaceutical interventions (NPIs). We illustrate the significance of the proposed closed-loop framework via extensive simulations and analysis of real, publicly-available testing data for COVID-19. Our results illustrate the significance of early testing and the emergence of a second wave of infections if NPIs are prematurely withdrawn.
\end{abstract}

\section{Introduction}
\label{section:introduction}

Mathematical modeling of infectious diseases that spread through the human population has a long history (see \cite{hethcote2000mathematics,pastor2015epidemic,draief2010epidemics,nowzari2016analysis,mei2017dynamics} for detailed surveys). One of the most fundamental mathematical models of epidemics is the susceptible-infected-recovered (SIR) epidemic where individuals can be in one of three possible compartments or states: susceptible, infected or recovered. Individuals who are susceptible get potentially infected by coming into contact with infected neighbors, while infected individuals recover at a certain rate. Individuals who recover develop immunity to the disease and do not become infected again. Since there is no strong evidence of widespread re-infection for at least a few months after recovery in the ongoing COVID-19 pandemic \cite{ota2020will}, the SIR epidemic and its variants have emerged as a popular framework to study its evolution \cite{pare2020modeling,giordano2020modelling}. Consequently, we focus on estimation and control of the SIR epidemic on networks in this work.  

As observed in case of COVID-19, infectious diseases, in the absence of appropriate medicine and vaccines at the early stages of an outbreak, non-pharmaceutical intervention (NPI) strategies must be deployed. The primary NPIs to control the spread of such epidemics include 
\begin{itemize}
\item reducing the interaction between nodes (for example, by restricting travel and imposing lock-down measures), and
\item deploying more resources (in terms of healthcare personnel, dedicated hospitals and medical equipment) thereby increasing the rate at which infected individuals get cured. 
\end{itemize}
Both of these interventions have a significant economic cost. Furthermore, in a networked setting, different nodes have different degrees of epidemic outbreak, and require different degrees of interventions.

\subsection{Related work and research gaps}

Following early works in \cite{wan2007network,wan2008designing,preciado2014optimal}, the existing literature on resource or NPI allocation for controlling the spread of epidemics has primarily focused on the susceptible-infected-susceptible (SIS) epidemic model on networks. In contrast with the SIR epidemic, under the SIS epidemic an individual after recovery can potentially be infected again if she/he comes in contact with other infected individuals. Nevertheless, for the SIS epidemic, there is a spectral condition which characterizes whether the disease persists in the population or if the proportion of infected population decays to zero.\footnote{\rev{Specifically, if the matrix that characterizes the linearized continuous-time SIS epidemic dynamic has all its eigenvalues in the left half of the complex plane, then the disease-free equilibrium is the unique asymptotically stable equilibrium of the epidemic. Otherwise, there is a unique nonzero endemic equilibrium of the dynamic which is locally asymptotically stable, and the disease-free equilibrium is an unstable equilibrium \cite{mei2017dynamics}.}} It was shown in \cite{preciado2014optimal} that the problem of optimal resource allocation to eradicate the disease is an instance of a {\it geometric program (GP)} \cite{boyd2007tutorial} which can be solved efficiently. This approach has been extended in several directions, such as to account for uncertainty in the network structure \cite{han2015data}, distributed algorithms \cite{ramirez2018distributed,mai2018distributed}, among others. However, the above optimization problems are solved off-line and do not use feedback to adapt the solution as the epidemic spreads.\footnote{Online approaches based on optimal control theory have been studied in \cite{eshghi2014optimal,zaman2008stability} to compute vaccination strategies, which are different from NPIs considered in this work.}

In contrast, investigations of optimal resource allocation to contain the spread of the SIR epidemic is particularly challenging because
\begin{itemize}
\item there is no known tractable characterization of the eventual number of recovered individuals for the SIR epidemic which can be minimized in an off-line manner,\footnote{Estimating the final size of the recovered population is challenging, and existing approaches rely on approximations that are not amenable for tractable optimization \cite{miller2012note}.} and 
\item the instantaneous growth rate $\lambda_{\max}(t)$ \rev{(defined in Section \ref{section:sir} and its connection with the reproduction number made precise in Remark \ref{remark:lambda_max})} depends on the proportion of susceptible individuals at each node which is time-varying and may not be accurately known. 
\end{itemize}
As a result, there have been few investigations on this problem until recently. In \cite{ogura2016efficient}, the authors showed that a GP can be formulated to minimize the expected cumulative number of people who get infected when the reproduction number is less than $1$. However, this condition is restrictive since for most epidemics the proportion of infected individuals shows an initial exponential increase before declining. 

\rev{With the emergence of COVID-19, extensions to the classical SIR model with additional infection states such as asymptomatic, hospitalized and quarantined \cite{giordano2020modelling,della2020network}, and a heterogeneous population differentiated in terms of age, underlying health conditions \cite{grundel2020much} are being explored. Several recent works have proposed nonlinear model predictive control (MPC) based approaches for NPI computation to minimize fatalities and/or to bound the infected proportion that requires hospitalization below a threshold, both in the single population \cite{kohler2020robust,morato2020optimal} as well as networked settings \cite{carli2020model,grundel2020much}. In a related work \cite{della2020network}, the authors highlight the importance of network structure dependent NPI allocation for effective containment of epidemics. However, the optimization problems formulated in the above MPC based approaches are multi-stage non-convex programs which are challenging to solve, and the complexity increases with the prediction horizon. Furthermore, performance measures such as minimizing the fatality numbers and constraints such as bounds on the number of severe cases are highly sensitive to certain model parameters which may not be accurately known and are likely time-varying.}

In this paper, we investigate a potential approach to minimize the reproduction number (formally defined in Section \ref{section:sir}) in an online manner. As mentioned above, this requires knowledge of the proportion of the susceptible subpopulation at each node of the network (as do MPC based approaches). As a result, the current proportions of susceptible, infected, and recovered subpopulations at different nodes and the parameters that govern the dynamics of the epidemic need to be learned from testing data (number of tests carried out, number of confirmed cases, and number of recoveries) that is made available every day by different jurisdictions.

Earlier work has studied the problem of learning the parameters of the epidemic dynamics from testing data in isolation. In particular, \cite{pare2018analysis} presents a data-driven framework for learning the parameters of the networked SIS epidemic. However, literature on SIR epidemic models have mostly focused on the scalar dynamics (without any network structure) with \cite{pare2020modeling} being a recent exception. In \cite{chen2020scenario}, least squares parameter identification was carried out assuming that the proportion of infected and recovered subpopulation and the daily change in the above (i.e., the state information) is proportional to the respective proportions in the testing data. Analogous assumptions were made in recent works on COVID-19 as well \cite{casella2020can,calafiore2020time}. In \cite{osthus2017forecasting,song2020epidemiological}, Bayesian Markov Chain Monte Carlo (MCMC) techniques were used for estimating the states and parameters. \rev{Even when the infected and susceptible proportions are perfectly known, uniquely identifying the parameters of the networked epidemic dynamic remains a challenging problem due to the structural properties of the epidemic models \cite{massonis2020structural,prasse2020network}.} Furthermore, the above models do not capture the following characteristics inherent in testing data.
\begin{itemize}
\item Due to limited testing capacity, high-risk (e.g., symptomatic or with travel history) individuals are more likely to get tested \cite{cohen2020countries}. 
\item For some diseases, such as COVID-19, individuals often show symptoms a few days after becoming infected (while they continue to infect others).
\end{itemize} 

\subsection{Proposed approach and contributions}
\label{section:overview}

In light of the above research gaps, we propose a closed-loop framework that integrates inferring the states of the epidemic from testing data, learning the parameters of the epidemic dynamics, and optimal resource allocation for the SIR epidemic on networks.\footnote{\rev{We focus on the SIR epidemic dynamic while introducing our framework as it is one of the most fundamental mathematical models of epidemic evolution. We briefly discuss (in Section VIII) how the proposed approach can be extended to account for exposed/asymptomatic compartments and other potential extensions.}}

In Section \ref{section:sir}, we state the discrete-time SIR epidemic model on a network and derive several properties of the state trajectories, including expanding the idea of the reproduction number\footnote{In epidemiology, the reproduction number is the number of infections one infection generates on average over the course of its infectious period. If less than one, the virus quickly dies out.} to the networked SIR model, and showing linear convergence to the equilibrium where no one is infected.

We then argue that one of the key requirements towards controlling the evolution of the infected population is to infer the current proportion of susceptible population from testing data (denoted by $\Omega(k)$). In Section \ref{subsection_testing_risk}, we propose a nonlinear observer model that relates testing data with the underlying states by incorporating the fact that (i) infected individuals are more likely to undergo testing than healthy individuals, and (ii) testing data is reflective of the epidemic state at a past time due to delay ($\tau$) in individuals developing symptoms. We then discuss how to infer the epidemic states from testing data in a Bayesian framework.

\tikzstyle{block} = [draw, fill=green!20, rectangle, align=center, minimum width=2.5em, minimum height=4em]
\tikzstyle{block1} = [draw, fill=blue!20, rectangle,  
    minimum height=1em, minimum width=1em]
\tikzstyle{sum} = [draw, fill=red!20, circle, node distance=1cm]
\tikzstyle{input} = [coordinate]
\tikzstyle{output} = [coordinate]
\tikzstyle{pinstyle} = [pin edge={to-,thick,black}]

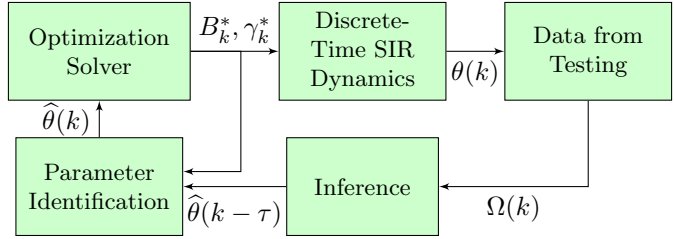
\begin{figure}
\centering
\footnotesize{
\begin{tikzpicture}[auto, node distance=3.5cm,>=latex']

\node [block, text width=2.2cm] (INC) {Optimization \\ Solver};
\node [block, right of=INC, text width=2cm, node distance=3.5cm] (OC) {Discrete-Time SIR Dynamics};
\node [block, below of=INC, text width=2cm, node distance=1.8cm] (SYS) {Parameter \\ Identification};
\node [block, text width=2cm, right of=OC, node distance=3cm] (B) {Data from \\ Testing};
\node [block, text width=1.8cm, below of=OC, node distance=1.8cm] (E) {Inference};

\draw [->] (INC) -- node[name=u1, midway, above] {\small{$B^*_k, \gamma^*_k$}} (OC);
\draw [->] (OC) -- node[name=u2, midway, below] {\small{$\theta(k)$}} (B);
\draw [->] (B) |- node[name=u3, near end, below] {\small{$\Omega(k)$}} (E);
\draw [->] (E) -- node[name=u5, midway, below] {\small{$\widehat{\theta}(k-\tau)$}} (SYS);
\draw [->] (INC) -- ($(INC.0) + (2/3,0)$) |- ($(SYS.0) + (0.2,0.2)$) -- ($(SYS.0) + (0,0.2)$);
\draw [->] (SYS) -- node[name=u4, midway, left] {\small{$\widehat{\theta}(k)$}} (INC);
\end{tikzpicture}}
\caption{Schematic of the proposed framework. Here $\theta(k)$ denotes the epidemic state at time $k$, $\Omega(k)$ denotes the testing data, $\widehat{\theta}(k)$ denotes the inferred epidemic states and $\tau$ denotes the delay factor as explained in Section \ref{section:overview}.}
\label{fig:loop}
\end{figure}

The delay factor $\tau$ necessitates predicting the current state of the epidemic ($\widehat{\theta}(k)$) from the past inferred values ($\widehat{\theta}(k-\tau)$). If the parameters of the epidemic dynamics are known for $\kappa \in [k-\tau,k]$, then the current state can be predicted from the dynamics. In the initial stages (before any interventions are deployed), these parameters may be learned from the available data. Eventually, once the optimal NPIs are deployed, they may  be used for predicting the current state. We present a least squares formulation to identify the virus spread parameters in Section \ref{section:least-square} in order to achieve this goal.

In Section \ref{section:optimization}, we formulate two complementary GPs to evaluate the trade-off between minimizing the cost of the NPIs and the growth-rate. The optimal solutions (denoted $B^*_k$, $\gamma^*_k$, to be made precise in the paper) are then deployed, and at the next iteration, new testing results are used to update the parameters and the predicted states, and the process repeats in an online manner. \rev{Although we do not solve a multi-stage optimization problem, a smaller growth rate, computed by solving the GP leads to the suppression of the infected proportion over a longer time-scale.} The schematic of the proposed scheme is shown in Fig.~\ref{fig:loop}. In Sections \ref{section:numerical} and \ref{section:real_empirical}, we illustrate the performance of the proposed approach via simulations and analysis of real testing data on COVID-19, respectively. Our results highlight the importance of early testing for accurate estimation and control performance, and the risk of a second wave of infection if NPIs are prematurely withdrawn. \rev{We conclude with a discussion on possible extensions and directions for future research in Section \ref{sec:conc}.}
\section{Discrete-Time SIR Epidemic Dynamic on Networks}
\label{section:sir}

In this section, we formally define the discrete-time SIR epidemic dynamic on networks. Let $G = (V,E)$ be a directed network or graph where $V$ is the set of nodes with $|V| = n$ and $E \subseteq V \times V$ is the set of edges. We consider a large-population regime where each node represents a subpopulation (such as a city or a county/district or a state) as opposed to a single individual. We denote the size of subpopulation~$i$ as $N_i$. \rev{Furthermore, we define $[n] := \{1,2,\ldots,n\}$.}

We denote by $\beta_{ij} \in \Rb_{\geq 0}$ the rate at which the infection can spread through the edge $(v_j,v_i) \in E$, \rev{directed from node $v_j$ to node $v_i$}, and by $\gamma_i \in \Rb_{>0}$ the rate at which an individual in subpopulation $i$ recovers from the infection. If two nodes $v_i$ and $v_j$ are not neighbors, then $\beta_{ij} = 0$. We assume that $\beta_{ii} \neq 0$ for every node $v_i$ since the individuals inside subpopulation~$i$ come in contact with each other. We define $V_i := \{j \in V | \beta_{ij} \neq 0\}$ to be the set of in-neighbors of node $v_i$. For better readability, we refer to the in-neighbors of a node as simply neighbors.\footnote{In the literature, the notation $a_{ij}$ is often used to denote the weight or contact pattern between nodes $i$ and $j$ and $\beta_i$ is used to denote the rate at which node $i$ gets infected when it comes in contact with its infected in-neighbors. The formulations are equivalent if we define $\beta_{ij} = \beta_i a_{ij}$.}

The proportions of the subpopulation at node $v_i$ that are susceptible, infected, and recovered at \rev{discrete time instant} $k$ are denoted by $s_i(k), x_i(k),$ and $r_i(k)$, respectively. Accordingly, we have $s_i(k), x_i(k), r_i(k) \in [0,1]$ and $s_i(k)+x_i(k)+r_i(k) = 1$ for all $k$ and $v_i \in V$. We now state the discrete-time evolution of the proportion of nodes in different epidemic states obtained via Euler discretization of the continuous-time SIR dynamic studied in \cite{mei2017dynamics}. Specifically, for a small enough sampling parameter $h > 0$, we have
\begin{subequations}
\begin{align} 
s_i(k+1) & = s_i(k) + h \big[-s_i(k) \sum^n_{j=1} \beta_{ij} x_j(k) \big], \label{eq:sir_si} \\
x_i(k+1) & = x_i(k) + h \big[s_i(k) \sum^n_{j=1} \beta_{ij} x_j(k) - \gamma_i x_i(k) \big], \label{eq:sir_xi}\\
r_i(k+1) & = r_i(k) + h \gamma_i x_i(k),
\label{eq:sir_ri}
\end{align}
\end{subequations}
\rev{where $k$ is discrete and corresponds to continuous time $t = hk$.} The initial conditions need to be specified such that $s_i(0), x_i(0), r_i(0) \in [0,1]$ and $s_i(0) + x_i(0) + r_i(0) = 1$ for every node $v_i$. In vector form the model becomes
\begin{subequations}\label{eq:sir_dt_vector_network}
\begin{align}
s(k+1) & = s(k) - h \diag(s(k)) B x(k), \\
x(k+1) & = x(k) + h \diag(s(k)) B x(k) -h \diag(\gamma) x(k) \nonumber \\ 
& = [\mathbf{I}_n + h \diag(s(k)) B - h \diag(\gamma)] x(k) \nonumber
\\ & =: A_k x(k), \label{eq:sir_dt_x}\\
r(k+1) & = \mathbf{1}_n - s(k+1) - x(k+1), 
\end{align}
\end{subequations}
where $B \in \Rb^{n \times n}$ is the matrix with $(i,j)$-th entry $\beta_{ij}$, $\diag(\gamma)$ is a diagonal matrix whose diagonal entries are the entries of the vector $\gamma$, $\mathbf{1}_n$ is the vector of dimension $n$ with all entries equal to $1$, and $\mathbf{I}_n$ is the identity matrix of dimension $n$. We now assume the following on the parameters such that the dynamic is well behaved.

\begin{assumption}\label{assm:1}
For all $i \in [n]$, let $0<h \gamma_i \leq 1$ and $0 < h \sum^n_{j=1} \beta_{ij} < 1$. The matrix $B$ is irreducible. Furthermore, $s_i(0) > 0$ for all $i\in[n]$.
\end{assumption}

Note that Assumption \ref{assm:1} is satisfied when the sampling parameter $h$ is chosen to be sufficiently small. If $h$ is not sufficiently small (i.e., sampling is infrequent), then it is possible for the states to become negative or exceed $1$, both of which are incompatible with their physical interpretations. The assumption also implies that $A_k$ is an irreducible non-negative matrix.\footnote{If $s_i(k) = 0$ for a node $v_i$ at time $k$, then $A_k$ may no longer be irreducible even if $B$ is irreducible.} Therefore, by the Perron-Frobenius Theorem for irreducible non-negative matrices \cite[Theorem~2.7 and Lemma~2.4]{varga}, $A_k$ has a positive real eigenvalue equal to its spectral radius, which, we denote by~$\lambda^{A_k}_{\max}$. 

We have the following result on the behavior of the discrete-time SIR dynamic under the above assumption. These are analogous to the behavior of the continuous-time dynamic \cite{mei2017dynamics}, but to the best of our knowledge, have not been formally proven in the literature. Our result also strengthens some of the observations in \cite{mei2017dynamics}.

\begin{theorem}\label{prop:main}
Consider the model in \eqref{eq:sir_dt_vector_network} under Assumption \ref{assm:1}. Suppose $s_i(0), x_i(0),$ $r_i(0) \in [0,1]$, $s_i(0) + x_i(0) + r_i(0) = 1$ for all $i \in [n] $ and $x_i(0) > 0$ for some~$i$. Then, for all $k \geq 0$ and $i \in [n]$, 
\begin{enumerate}
\item[1)] $s_i(k), x_i(k), r_i(k) \in [0,1]$ and $s_i(k) + x_i(k) + r_i(k) = 1$,
\item[2)] $s_i(k+1) \leq s_i(k)$, 
\item[3)] $\lim_{k \to \infty} x_i(k) = 0$ for $i \in [n]$,
\item[4)] $\lambda^{A_k}_{\max}$ is monotonically decreasing as a function of $k$,
\item[5)] there exists $\bar{k}$ such that $\lambda^{A_{k}}_{\max}< 1$ for all $k \geq \bar{k}$, and
\item[6)] there exists $\bar{k}$, such that $x_i(k)$ converges linearly\footnote{\rev{Recall the definition that a sequence $\{x(k)\}_{k \geq 0}$ converges linearly to $0$ if
\begin{equation*}
    \lim_{k \to \infty} \frac{\|x(k+1)\|}{\|x(k)\|} <1.
\end{equation*}
}} to $0$ for all~$k \geq~\bar{k}$, $i \in [n]$. 
\end{enumerate}
\end{theorem}
\begin{proof} See Appendix \ref{section:proof}.
\end{proof}

Theorem~\ref{prop:main} shows that the model in \eqref{eq:sir_dt_vector_network} is well-defined, the susceptible proportions and the growth rate decrease monotonically over time, the growth rate will eventually be less than $1$, and the infected proportion will go to $0$, in at least linear time for large enough $k$. 

\begin{remark}\label{remark:lambda_max}
The largest eigenvalue $\lambda^{A_k}_{\max}$ is a generalization of the reproduction number to the networked epidemic setting, that is, if $\lambda^{A_k}_{\max} < 1$ the virus quickly dies out. Our result rigorously proves the observation that for epidemics that follow an SIR-type dynamic, such as COVID-19, the reproduction number eventually falls below $1$.
\end{remark}

We use these results as the baseline for the control techniques presented in Section \ref{section:optimization}, which also require inferring the states from the testing data (Section \ref{subsection_testing_risk}) and estimating the spread parameters for forecasting the states (Section \ref{section:least-square}).
\section{Nonlinear Observer Model and Inference from Testing Data}\label{subsection_testing_risk}

As discussed earlier, part of the challenge in controlling the spread of infectious diseases such as COVID-19 is that the true prevalence (i.e., the underlying state) of the disease in a given population is not known. Individuals need to be tested in order to determine if they are infected by the disease or pathogen under consideration. At early stages of the epidemic, many countries and regions do not have enough capacity to test a large number of people since their testing kits are limited. As a result, testing is conducted on individuals who show symptoms (such as fever and shortness of breath associated with COVID-19). However, drawing inference about the underlying spread of the epidemic from such testing data is not straightforward since (i) a large fraction of infected and contagious individuals never show any symptoms, and (ii) similar symptoms are also exhibited by patients who suffer from other related illnesses \cite{hu2020clinical}. Furthermore, testing data on a given day reveals individuals who became infected at least a few days earlier (as opposed to information about the new infections on that day).

\sloppy
In this section, we develop a nonlinear observer model that relates the testing data or observations with the underlying epidemic states ($s(k), x(k), r(k)$). We then propose a Bayesian approach to infer the underlying states from observed testing data. We denote the inferred quantities with the symbol $\ \widehat{\cdot} \ $. 

\subsection{Nonlinear Observer Model of Testing Data}
\label{section:observer}

\rev{The testing data at time $k$ at node $v_i$ is denoted by $\Omega_i(k) = (z_i(k),c_i(k),d_i(k))$ where $z_i(k)$ denotes the number of tests carried out, $c_i(k)$ denotes the number of new confirmed cases and $d_i(k)$ denotes the number of new removed cases (sum of number of new recoveries and new death counts).} We denote the cumulative number of confirmed and removed cases by $\mathtt{C}_i(k):=\sum_{l=0}^k c_i(l)$ and $\mathtt{D}_i(k):=\sum_{l=0}^k d_i(l)$, respectively. Accordingly, the number of known active cases is given by $\mathtt{A}_i(k) := \mathtt{C}_i(k) - \mathtt{D}_i(k)$ and the \textit{daily change} in the number of known/active cases is given by $a_i(k) = c_i(k) - d_i(k)$.

Note that the known active cases are analogous to the proportion of infected individuals. Therefore, it is reasonable to treat the proportion of new infections at time $k$ (i.e., $\frac{c_i(k)}{z_i(k)}$) or the {\it test positivity rate} as representative of the proportion of new infections at node $v_i$. However, as discussed earlier, the confirmed cases at time $k$ are often found to have caught the infection several days prior to being tested, i.e., with a certain delay denoted by $\tau \geq 1$. Therefore, we assume $\frac{c_i(k)}{z_i(k)}$ to be representative of the decrease in the proportion of susceptible individuals at time $k - \tau$, denoted by $-\Delta s_i(k-\tau)$, \rev{where $$\Delta s_i(k) := s_i(k) - s_i(k-1) = -h s_i(k-1) \sum^n_{j=1} \beta_{ij} x_j(k-1).$$ The quantities $\Delta x_i(k)$ and $\Delta r_i(k)$ are defined in an analogous manner.} While most of the prior work assumes $\frac{c_i(k)}{z_i(k)}$ to be proportional to $x_i(k)$ and/or $-\Delta s_i(k)$ \rev{\cite{chen2020scenario, casella2020can}}, we propose a Bayesian framework to model the fact that testing strategies are not uniform, and formally relate $\frac{c_i(k)}{z_i(k)}$ to $\Delta s_i(k-\tau)$.

We assume that a proportion $h_i(k)$ of the population at node $v_i$ \rev{exhibit a higher risk of being infected, potentially because they either exhibit symptoms associated with the disease, are known contacts of confirmed cases, have a travel history in affected regions, or any combinations thereof. For the sake of brevity, we refer this proportion as the {\it high risk} category. Similarly, $l_i(k) := 1 - h_i(k)$ denotes the proportion of population at node $v_i$ who belong to the {\it low risk} category.} 

Now, let $H_i(k)$ be a random variable with $H_i(k) = 1$ (resp. $H_i(k) = 0$) if a randomly \rev{chosen individual of the population in node $v_i$} belongs to the high risk (resp. low risk) category. Accordingly, $\Pb(H_i(k) = 1) = h_i(k)$. We also define a $\{0,1\}$-valued random variable $DX_i(k)$ with $DX_i(k) = 1$ if a randomly chosen individual became infected at time $k$ and $DX_i(k) = 0$, otherwise. We now introduce the following notation to denote certain conditional probabilities of interest. 

In particular, we define
\begin{align*}
& p_{hx, i}(k) := \Pb(H_i(k) = 1 | DX_i(k-\tau) = 1) 
\\ \implies & \Pb(H_i(k) = 0 | DX_i(k-\tau) = 1) = 1- p_{hx, i}(k), \\
& p_{hh, i}(k) := \Pb(H_i(k) = 1 | DX_i(k-\tau) = 0) 
\\ \implies & \Pb(H_i(k) = 0 | DX_i(k-\tau) = 0) = 1- p_{hh, i}(k),
\end{align*}
where $p_{hx, i}(k), p_{hh, i}(k) \in [0,1]$. \rev{The parameter $p_{hx, i}(k)$ denotes the proportion of individuals who became infected at $k-\tau$ and belong to the high risk category at time $k$, i.e., the proportion of infected individuals who show symptoms within $\tau$ time steps of being infected. Similarly, $1-p_{hx, i}(k)$ captures the proportion of infected individuals who remain asymptomatic after $\tau$ steps of becoming infected. The parameter $p_{hx, i}(k)$ depends on the characteristics of the epidemic and the population (e.g., age, prevalence of comorbidity) at node $v_i$.} The parameter $p_{hh, i}(k)$ captures the proportion of individuals who belong to the high risk category, but did not become infected $\tau$ time steps earlier.  

We now apply Bayes' law to compute the probability of a randomly chosen high risk individual being infected $\tau$ time steps earlier. \rev{For better readability, we define $\bar{k} := k - \tau$.} Specifically, we compute
\begin{align}
\Pb(DX_i(\bar{k}) = 1 | H_i(k) = 1) & = \frac{\Pb(H_i(k) = 1 | DX_i(\bar{k}) = 1) \Pb(DX_i(\bar{k}) = 1)}{\Pb(H_i(k) = 1)} \nonumber
\\ & \quad = \frac{p_{hx,i}(k)(-\Delta s_i(\bar{k}))}{p_{hx, i}(k)(-\Delta s_i(\bar{k})) + p_{hh, i}(k)(1-(-\Delta s_i(\bar{k})))} \nonumber
\\ & \quad =: p_{xh,i}(k). \label{eq:pxhi_bayes}
\end{align}
Similarly, the probability of a randomly chosen low risk individual being infected $\tau$ time steps earlier is
\begin{align}
\Pb(DX_i(\bar{k}) = 1 | H_i(k) = 0) & = \frac{\Pb(H_i(k) = 0 | DX_i(\bar{k}) = 1) \Pb(DX_i(\bar{k}) = 1)}{\Pb(H_i(k) = 0)} \nonumber
\\ & \quad = \frac{-\Delta s_i(\bar{k})(1-p_{hx, i}(k))}{-\Delta s_i(\bar{k})(1-p_{hx, i}(k)) + (1-p_{hh, i}(k))(1+\Delta s_i(\bar{k})))} \nonumber
\\ & \quad =: p_{xl,i}(k). \label{eq:pxli_bayes}
\end{align}
Furthermore, both $p_{xh,i}(k)$ and $p_{xl,i}(k)$ are monotonically increasing in $-\Delta s_i(\bar{k})$.

We now relate the number of tests and confirmed positive cases with the underlying states. At time $k$, suppose the authorities at node $v_i$ decide to carry out $z_{h, i}(k) \in [0,z_i(k)]$ number of tests on high risk individuals and $z_{l, i}(k) = z_i(k) - z_{h, i}(k)$ number tests on low risk individuals. Assuming that testing is accurate (with false positive and false negative rates being $0$),\footnote{Modifying the proposed Bayesian approach to incorporate inaccuracy in testing data is beyond the scope of this paper and will be explored in follow up work.} the confirmed cases among those tested is the sum of the number of confirmed cases among tested high risk individuals and among tested low risk individuals. Accordingly, we model
\begin{align}
c_i(k) & \sim \mathtt{Bin}(z_{h, i}(k),p_{xh,i}(k)) + \mathtt{Bin}(z_{l, i}(k),p_{xl,i}(k)), \label{eq:cik_bin}
\end{align}
where $\mathtt{Bin}(n,p)$ denotes the Binomial distribution with parameters $n$ and $p$. 

We now relate the observed number of recoveries ($d_i(k)$) with the underlying states in an analogous manner. Recall that under the SIR epidemic dynamics, the change in the proportion of recovered individuals, $\Delta r_i(k) = r_i(k) - r_i(k-1) = h \gamma_i x_i(k-1)$. In the observed data, the quantity analogous to the change in the fraction of recovered individuals is $d_i(k)$, which denotes the number of known removed cases among the known active cases $\mathtt{A}_i(k-1)$. Accordingly, we assume
\begin{equation}\label{eq:nonlin_obs_d}
d_i(k) \sim \mathtt{Bin}(\mathtt{A}_i(k-1), h\gamma_i).
\end{equation}
In other words, each known active case recovers with probability \rev{$h\gamma_i$}. When the number of active cases is large, $d_i(k)$ is approximately equal to $h\gamma_i \mathtt{A}_i(k-1)$.  


\subsection{Inference from Testing Data}

The above analysis formally relates the observed quantities with the underlying epidemic states. If the parameters $p_{hx, i}(k)$ and $p_{hh, i}(k)$ and the number of confirmed cases among the tested high and low risk populations are known, then maximum likelihood inference of $-\Delta s_i(k-\tau)$ can be computed without much difficulty as both $p_{xh, i}(k)$ and $p_{xl, i}(k)$ are monotonically increasing in $-\Delta s_i(k-\tau)$. 

Nevertheless, in the rest of the section, as well as in our numerical results, we focus on the practically motivated special case where only high risk individuals undergo testing, i.e., $z_{h, i}(k) = z_i(k)$. Note that this has been the practice in many jurisdictions in the world during the ongoing COVID-19 pandemic \cite{cohen2020countries}. Following \eqref{eq:pxhi_bayes}, we have
\begin{align}
\frac{c_i(k)}{z_i(k)} = \widehat{p_{xh, i}}(k) & = \left[ 1 + \frac{p_{hh, i}(k)}{p_{hx, i}(k)}  \bigg[\frac{1}{-\widehat{\Delta s}_{i}(k-\tau)}-1\bigg] \right]^{-1} \nonumber 
\\ & = \left[ 1 + \frac{1}{\alpha_i(k)}  \bigg[\frac{1}{-\widehat{\Delta s}_{i}(k-\tau)}-1\bigg] \right]^{-1}, \label{eq:state2output}
\end{align}
where $\alpha_i(k) := \frac{p_{hx, i}(k)}{p_{hh, i}(k)}$, i.e., the ratio of the probability of an infected individual being high risk (or undergoing testing) and the probability of a healthy individual being high risk (or undergoing testing). \rev{In practice, it is more likely that an infected individual undergoes testing, and as a result, it is reasonable to assume that $\alpha_i(k) \geq 1$. In addition, $\alpha_i(k)$ potentially depends on the number of tests carried out in a given day. If fewer tests are carried out, it is likely that these tests are conducted on individuals who are sick or are isolated through contact tracing. As a result, $\alpha_i(k)$ is potentially large when $z_i(k)$ is relatively small, and vice versa.} Note that if $\alpha_i(k)=1$, i.e., a healthy person is equally likely to undergo testing compared to an infected person, then we have $-\widehat{\Delta s}_{i}(k-\tau)  = \frac{c_i(k)}{z_i(k)}$.


Given the testing data and assuming that we know $\alpha_i(k)$, we now discuss how to infer the underlying states. Suppose the daily testing data $\Omega_i(k)$ is available to us over the time interval $k \in [T_1+\tau, T_2+{\tau}]$ at node $v_i$. 

\rev{Following \eqref{eq:state2output},} we define the inferred fraction of new infections at node $v_i$ as 
\begin{align}
-\widehat{\Delta s}_{i}(k) & := \Big[ 1 - \alpha_i(k+\tau) + \alpha_i(k+\tau) \frac{z_i(k+\tau)}{c_i(k+\tau)}\Big]^{-1}, \label{eq:inf_inf_risk}
\end{align}
for $k \in [T_1, T_2]$. Similarly, following \eqref{eq:nonlin_obs_d}, the change in the proportion of recovered individuals is 
\begin{equation}\label{eq:inf_deltar}
\widehat{\Delta r}_{i}(k) := \frac{d_i(k)}{\mathtt{A}_i(k-1)} \widehat{x}_i(k-1) = \frac{d_i(k) \widehat{x}_i(k-1)}{\mathtt{C}_i(k-1)-\mathtt{D}_i(k-1)},
\end{equation}
for $k \in [T_1+\tau,T_2]$. In the pathological case with $c_i(k) = 0$ and $\mathtt{A}_i(k-1) = 0$, we assume that $\widehat{\Delta s}_{i}(k) = 0$ and $\widehat{\Delta r}_{i}(k) = 0$, respectively. \rev{Note that the inferred quantities defined in \eqref{eq:inf_inf_risk} and \eqref{eq:inf_deltar} are accurate only when $c_i(k+\tau)$, $z_i(k+\tau)$, and $\mathtt{A}_i(k-1)$ are sufficiently large.} \rev{With the change in inferred proportions stated in \eqref{eq:inf_inf_risk} and \eqref{eq:inf_deltar}, we define the inferred states recursively for $k \in [T_1,T_2]$ as:
\begin{align*}
\widehat{s}_i(k) & = \widehat{s}_i(k-1) + \widehat{\Delta s}_{i}(k), 
\\ \widehat{x}_i(k) & = \widehat{x}_i(k-1) - \widehat{\Delta s}_{i}(k) - \widehat{\Delta r}_{i}(k) = \left[1-\frac{d_i(k)}{\mathtt{A}_i(k-1)}\right] \widehat{x}_i(k-1) - \widehat{\Delta s}_{i}(k),
\end{align*}
with $\widehat{s}_i(T_1-1)$ and $\widehat{x}_i(T_1-1)$ being the initial inferred quantities. Note that for $k < T_1 + \tau$, we have $d_i(k) = 0$, which implies $\widehat{\Delta r_i}(k) = 0$ and hence $\widehat{x}_i(k) = \widehat{x}_i(k-1) - \widehat{\Delta s}_{i}(k)$.}

\rev{We can equivalently state the inferred state at time $k \in [T_1,T_2]$ in terms of the initial inferred quantities as:
\begin{align}
\widehat{s}_i(k) & = \widehat{s}_i(T_1-1) + \sum^k_{j = T_1} \widehat{\Delta s}_{i}(j), \label{eq:inf_s_final}
\\ \widehat{x}_i(k) & = \widehat{x}_i(T_1-1) \prod^k_{j=T_1} \left[1-\frac{d_i(j)}{\mathtt{A}_i(j-1)}\right] - \widehat{\Delta s}_{i}(k) - \sum^{k-1}_{j=T_1} \left[ \widehat{\Delta s}_{i}(j) \prod^k_{l=j+1} \left( 1 - \frac{d_i(l)}{\mathtt{A}_i(l-1)} \right) \right]. \label{eq:inf_x_final}
\end{align}
In other words, the inferred trajectory of susceptible and infected proportions are linear functions of the initial inferred proportions.}

\rev{Note further that the inferred quantities are nonlinear in the $\alpha_i(k)$ parameters (via $\widehat{\Delta s}_{i}(k)$ defined in \eqref{eq:inf_inf_risk}). In the following section, we first present a least squares optimization problem to estimate the initial inferred state and the parameters of the SIR epidemic model assuming $\alpha_i(k)$ parameters are known. Subsequently, we discuss how to possibly infer $\alpha_i(k)$ parameters from testing data.}


\section{Epidemic Parameter Identification and Forecasting}
\label{section:least-square}

Recall from prior discussion that optimal resource allocation to minimize the growth rate of the disease, or the cost of NPIs, for the SIR epidemic, requires knowledge of the current proportion of susceptible individuals at each node of the network. However, due to the delay $\tau$, testing data up to time $k$ only suffice to infer the epidemic states up to $k-\tau$. As a result, we need to predict the current proportion of susceptible individuals using the inferred states up to $k-\tau$. Prediction, or forecasting the future trajectory, of the epidemic is also of independent interest (beyond the optimal resource allocation problem). If the parameters of the epidemic dynamics are known, then those can be used to predict the current state of the epidemic from the inferred states up to $k-\tau$. 

However, at the early stages of the epidemic, there is substantial uncertainty regarding the values of $\beta_{ij}$ and $\gamma_i$ parameters. \rev{In this section, we extend the least squares estimation technique proposed in \cite{chen2020scenario,prasse2020network} to learn the unknown $\beta_{ij}$ and $\gamma_i$ parameters in the networked setting. First we consider the case when the initial states are assumed to be known.}

\subsection{Estimation with given initial inferred states}

\rev{Recall from the above discussion that if the initial inferred states are known, then the inferred states are uniquely determined for $k \in [T_1,T_2]$ following \eqref{eq:inf_s_final} and \eqref{eq:inf_x_final}. If, for instance, testing data is available from the onset of the epidemic at node $v_i$, i.e., for the initial few time instances, there have been no infections and recoveries, then it is reasonable to assume that the initial proportion of infected individuals is $0$ and susceptible individuals is $1$.}

\rev{We now formulate a suitable least squares problem to learn the infection and recovery rate parameters from inferred states. Recall from the definition of the SIR epidemic dynamic in \eqref{eq:sir_si} that $\Delta s_i(k+1) = h \big[-s_i(k) \sum^n_{j=1} \beta_{ij} x_j(k)\big]$. Note further that if the network topology, i.e., the presence or absence of edges between different nodes are known, then we have $\beta_{ij} = 0$ if $(v_j,v_i) \notin E$. As a result, the dimension of the infection rate parameters is $\Rb^{|E|}$. Similarly, according to \eqref{eq:sir_ri}, $\Delta r_i(k+1) = h \gamma_i x_i(k)$. We use the inferred state information in the above relationships to learn the infection and recovery rate parameters $\{\{\beta_{ij}\}_{(i,j) \in E}, \{\gamma_i\}_{i \in [n]}\} \in \Rb^{n+|E|}_{+}$ by solving:
\begin{align}
\min_{\gamma \in \Rb^n_{+}, \beta \in \Rb^{|E|}_{+}} & \, \, \sum^n_{i=1} \sum^{T_2}_{k=T_1}  \left(\! 1 + h \frac{\widehat{s}_i(k-1)}{\widehat{\Delta s}_{i}(k)} \sum_{j \in V_i} \beta_{ij}  \widehat{x}_j(k-1) \right)^2 \nonumber
\\ & + \quad \!\! \sum^n_{i=1} \sum^{T_2}_{k=T_1+\tau} \! \! \left(1 - h \gamma_i \frac{\widehat{x}_i(k-1)}{\widehat{\Delta r}_{i}(k)}\right)^2. \label{eq:ls_opt}
\end{align}}

\rev{Since all the inferred quantities are known, the above is a linear regression problem with the cost function being convex in the decision variables, and thus, the problem can be solved efficiently. In \eqref{eq:ls_opt}, we set the response variable of the regression problem to $1$ by normalizing the inferred states with change in inferred states for better numerical conditioning.} If additional information such as bounds on the $\beta_{ij}$ parameters are available, those may be incorporated as constraints in \eqref{eq:ls_opt}. We denote the optimal solutions of \eqref{eq:ls_opt} as $\widehat{\gamma}_i$ and $\widehat{\beta}_{ij}$.

\subsection{Estimating parameters as well as initial inferred states}
\label{section:least-squareb}

\rev{While the optimization problem in \eqref{eq:ls_opt} is convex, the predicted states (using learned $\widehat{\gamma}_i$ and $\widehat{\beta}_{ij}$) will likely be very different from true states if the assumed initial inferred states are erroneous. It is indeed the case when testing is carried out after the disease has already spread for some time.}

\rev{In such cases, we propose to treat the initial inferred states ($\widehat{s}_i(T_1-1), \widehat{x}_i(T_1-1)$) as decision variables. For brevity of notation, let $\tilde{s}^0_i = \widehat{s}_i(T_1-1)$ and $\tilde{x}^0_i = \widehat{x}_i(T_1-1)$. Let the inferred states given by \eqref{eq:inf_s_final} and \eqref{eq:inf_x_final} with $\tilde{s}^0_i$ and $\tilde{x}^0_i$ being initial conditions be denoted by $\widehat{s}_i(k)(\tilde{s}^0_i)$ and $\widehat{x}_i(k)(\tilde{x}^0_i)$ with a slight abuse of notation. We now solve the following problem:
\begin{align}
\min_{\substack{\tilde{s}^0, \tilde{x}^0, \gamma \in \Rb^n_{+} \\ \beta \in \Rb^{|E|}_{+}}} & \, \, \sum^n_{i=1} \sum^{T_2}_{k=T_1}  \left(\! 1 + h \frac{\widehat{s}_i(k-1)(\tilde{s}^0_i)}{\widehat{\Delta s}_{i}(k)} \sum_{j \in V_i} \beta_{ij}  \widehat{x}_j(k-1)(\tilde{x}^0_i) \right)^2 \nonumber
\\ & + \!\! \sum^n_{i=1} \sum^{T_2}_{k=T_1+\tau} \! \! \left(1 - h \gamma_i \frac{\widehat{x}_i(k-1)(\tilde{x}^0_i)}{\widehat{\Delta r}_{i}(k)}\right)^2 + w \!\! \sum^n_{i=1} (\tilde{s}^0_i-1)^2 \nonumber
\\ \text{s.t.} & \quad  \tilde{s}^0_i \in [0,1], \tilde{x}^0_i \in [0,1], \forall i \in [n],  \label{eq:ls_opt_new}
\\ & \quad \widehat{s}_i(k)(\tilde{s}^0_i), \widehat{x}_i(k)(\tilde{x}^0_i) \in [0,1], \forall i \in [n], k \in [T_1,T_2], \nonumber
\\ & \quad \widehat{s}_i(k)(\tilde{s}^0_i) + \widehat{x}_i(k)(\tilde{x}^0_i) \leq 1, \forall i \in [n], k \in [T_1,T_2], \nonumber
\end{align}
where the last three constraints guarantee that the state trajectories from the initial inferred states following \eqref{eq:inf_s_final} and \eqref{eq:inf_x_final} are consistent with their physical interpretations. These constraints are affine in the decision variables. The last term in the cost function is a penalty term that keeps the initial proportion of susceptible individuals close to $1$ as is often the case. The parameter $w$ is a weighting factor.} 

\rev{In contrast with \eqref{eq:ls_opt}, the cost function in \eqref{eq:ls_opt_new} is non-convex due to the product terms $\widehat{s}_i(k-1)(\tilde{s}^0_i)$ and $\widehat{x}_j(k-1)(\tilde{x}^0_i)$ both of which are affine functions of $\tilde{s}^0_i$ and $\tilde{x}^0_i$, respectively. Therefore, problem \eqref{eq:ls_opt_new} is often computationally challenging to solve. However, it offers significant benefits in practice as it solves for the initial inferred states as well as parameters that lead to feasible state trajectories.}

Both \eqref{eq:ls_opt} and \eqref{eq:ls_opt_new} assume that the parameters $\{\{\gamma_i\}_{i \in [n]}, \{\beta_{ij}\}_{(i,j) \in E}\} \in \Rb^{n+|E|}_{+}$ do not significantly change over the interval $[T_1,T_2]$. If the parameters change due to NPIs by the authorities, then the above formulation can be suitably modified to learn the epidemic parameters both before and after the imposition of NPIs by considering suitable sub-intervals during which different NPIs were put into place. Other approaches, such as \cite{chen2020scenario}, define $\beta_{ij}$'s to be parametric functions of NPIs. \rev{These settings can also be handled by suitably modifying the above formulation. For longer time intervals, parameter estimation can be carried out over a moving horizon of suitable length.}

\subsection{Learning the $\alpha$ hyper-parameter}
\label{sec:alpha_tuning}

\rev{The discussion on inference and parameter estimation thus far has assumed that the hyper-parameters $\alpha_i(k)$ are known. Recall that $\alpha_i(k)$ captures the ratio of the probability of an infected individual undergoing testing and a healthy individual undergoing testing.}

\rev{Note from \eqref{eq:inf_inf_risk} that the inferred proportion of new infections, $-\widehat{\Delta s}_i(k)$, is inversely proportional to $\alpha_i(k+\tau)$. Thus, if our conjectured value of $\alpha$ is much smaller than the true $\alpha$ value, then $-\widehat{\Delta s}_i(k)$ would be much larger than its true value. This would lead to $\widehat{s_i}(k)$ becoming negative for some $k \in [T_1,T_2]$. In contrast, if the conjectured value of $\alpha$ is much larger than the true $\alpha$ value, then the inferred states would be much smaller, possibly leading to total number of new infections being smaller than the known confirmed cases, i.e., $-\widehat{\Delta s}_i(k)N_i < c_i(k)$, where $N_i$ is the population size at node $v_i$. Thus, it is not difficult to rule out extreme values of $\alpha$.}

\rev{Once a sensible range for this parameter is determined, we propose to solve the optimization problem in \eqref{eq:ls_opt_new} (or \eqref{eq:ls_opt}) for different values of $\alpha$ over this range at suitable granularity. The value of $\alpha$ which leads to the smallest optimal cost is likely to be close to its true underlying value. In Section \ref{section:numerical}, we numerically evaluate the above approach for different values of $\alpha$ (assuming that all nodes have the same value of $\alpha$) and different network sizes, and show that the value of $\alpha$ that leads to the smallest optimal cost is almost always close to the true value of $\alpha$ that is used to generate the testing data.}

\subsection{Forecasting}

The learned parameters of the epidemic dynamics from \rev{\eqref{eq:ls_opt} or \eqref{eq:ls_opt_new}, denoted $\widehat{\gamma}_i$ and $\widehat{\beta}_{ij}$, together with appropriate initial inferred states}, can now be used to predict the current and future state of the epidemic using the inferred states as the initial conditions. In order to avoid introducing additional notation, we use $\ \widehat{\cdot} \ $ to also denote the predicted states. Specifically, for $k \geq T_1$, we compute the {\it predicted states} as
\begin{subequations}\label{eq:for}
    \begin{align}
        \widehat{s}(k+1) &= \widehat{s}(k) - h \diag(\widehat{s}(k)) \widehat{B} \widehat{x}(k), \label{eq:for_s}\\
        \widehat{x}(k+1) &=\widehat{x}(k) + h \diag(\widehat{s}(k)) \widehat{B} \widehat{x}(k) - h \diag(\widehat{\gamma}) \widehat{x}(k),
    \end{align}
\end{subequations}
where $\widehat{B}$ is the matrix with $(i,j)$-th entry $\widehat{\beta}_{ij}$ and $\diag(\widehat{\gamma})$ is a diagonal matrix with diagonal entries being $\widehat{\gamma}_i$. 

\begin{remark}\label{remark:prasse}
\rev{In a recent work \cite{prasse2020network}, the authors show that it is often not possible to learn the true values of the infection rate parameters ($\beta_{ij}$'s) by solving \eqref{eq:ls_opt} even when the inferred states are accurate due to the structure of the regression problem. Nevertheless, the predicted states using the (incorrect) learned values of the parameters closely align with the true state trajectory. Our empirical results are consistent with the above observation by \cite{prasse2020network}.}
\end{remark}

\begin{remark} 
The formulation in \eqref{eq:ls_opt} (and \eqref{eq:ls_opt_new}) admits a separable structure and as a result, each node $v_i$ can learn their respective $\beta_{ij}$ and $\gamma_i$ values using local information from their own testing data and obtaining the estimates $\widehat{x}_j(k)$ from their neighbors \rev{(assuming that all nodes are aware of their respective $\alpha$ values). However, the hyper-parameter tuning at each node depends on the accuracy of estimates received from the neighboring nodes which makes the problem challenging. We motivate this problem as a promising avenue for future research in Section \ref{sec:conc}.}
\end{remark}

In the following section, we formulate optimization problems to compute optimal NPIs to control the spread of the epidemic using the predicted state from \eqref{eq:for}.


\section{Optimal Resource Allocation via Geometric Programming}\label{section:optimization}

Recall from Theorem~\ref{prop:main} that the growth rate of the infected fraction of the population at time $k$ is given by
$\lambda^{A_k}_{\max}$ (the largest eigenvalue of the matrix $\mathbf{I}_n + h \diag(s(k)) B - h \diag(\gamma)$). Since $A_k$ is a function of $s(k)$, in the absence of perfect knowledge of $s(k)$, we use the inferred/predicted value of $s(k)$, denoted by $\widehat{s}(k)$ (given in \eqref{eq:for_s}). The goal of the social planner is to control the spread of the epidemic by choosing the recovery or curing rates, i.e., the $\gamma_i$ parameters (which, for instance, correspond to deploying a larger number of healthcare personnel and/or medical equipment) and the infection or contact rates, i.e., the $\beta_{ij}$ parameters (which correspond to imposing social distancing or lock-down measures).

We present two geometric programming formulations to aid the social planner's decision-making with regards to NPIs. First, we consider the problem of minimizing the instantaneous growth rate $\lambda^{A_k}_{\max}$ by optimally allocating the NPIs subject to budget constraints. Following analogous arguments in \cite{han2015data,preciado2014optimal}, this problem is equivalent to a GP given by:
\begin{subequations}\label{eq:GP_SIR}
\begin{align}
\min_{\lambda, \bar{\gamma}, w, \beta} & \, \, \lambda \\
\text{s.t. } & \, \, \sum_{j \in V_i} h \widehat{s}_i(k) \beta_{ij} w_j/w_i + \bar{\gamma}_i \leq \lambda \quad \forall i \in [n], \label{eq:GP_SIR_sep} \\
& \, \, \sum_{(v_i,v_j) \in E} f_{ij}(\beta_{ij}) \leq C_1, \label{eq:GP_SIR_bud} \\
& \, \, \sum^n_{i=1} g_i(\bar{\gamma}_i) \leq C_2, \label{eq:GP_SIR_bud1} \\
& \, \, \bar{\gamma}_l \leq \bar{\gamma} \leq \bar{\gamma}_u, \quad \beta_l \leq \beta \leq \beta_u, \label{eq:GP_SIR_bound2} \\
& \lambda \in \Rb_{+}, \bar{\gamma} \in \Rb^n_{+}, w \in \Rb^n_{+}, \beta \in \Rb^{|E|}_{+},
\end{align}
\end{subequations}
where the new variable $\bar{\gamma}_i := 1-h\gamma_i$ is used so that the constraints in \eqref{eq:GP_SIR_sep} remain posynomials. The constraints in \eqref{eq:GP_SIR_bound2} correspond to bounds on the infection and recovery rates \rev{where $\beta_u, \beta_l \in \Rb^{|E|}_{+}$ denote the upper and lower bounds on the infection rates on the edges of the network and $\bar{\gamma}_u, \bar{\gamma}_l \in \Rb^n_{+}$ denote the upper and lower bounds on $\bar{\gamma}$, respectively.} The functions $f_{ij}$ and $g_i$ are posynomial cost functions of NPIs \rev{that are non-increasing} in the arguments, and the constraints in \eqref{eq:GP_SIR_bud} and \eqref{eq:GP_SIR_bud1} are budget constraints on NPIs with $C_1$ and $C_2$ being the budgets for infection rates and recovery rates, respectively. 

The optimal value $\lambda^*_k$ corresponds to the largest eigenvalue of $A_k$ with parameters $\gamma^*_k$ and $B^*_k$; the latter denote the optimal recovery rates and infection rate matrix, respectively. The optimal $w^*_k$ is the eigenvector corresponding to $\lambda^*_k$. Furthermore, $\lambda^*_k$ is the smallest growth rate that can be achieved given the budget constraints. The above problem is solved repeatedly in an online manner. At time step $k+1$, we again obtain $\widehat{s}(k+1)$ via feedback, and solve \eqref{eq:GP_SIR} with $A_{k}$ replaced by $A_{k+1}$. 

A problem complementary to \eqref{eq:GP_SIR} is to minimize the cost of NPIs subject to the constraint that the growth rate is bounded by $\lambda_k$. This problem is given by:
\begin{subequations}\label{eq:GP_SIR2}
\begin{align}
\min_{\bar{\gamma}, w, \beta} & \, \, \Psi(\beta,\bar{\gamma}) :=\sum_{(v_i,v_j) \in E} f_{ij}(\beta_{ij}) + \sum^n_{i=1} g_i(\bar{\gamma}_i) \\
\text{s.t. } & \, \, \sum_{j \in V_i} h \widehat{s}_i(k) \beta_{ij} w_j/w_i + \bar{\gamma}_i \leq \lambda_k \quad \forall i \in [n], \label{eq:GP_SIR2_sep} \\
& \, \, \bar{\gamma}_l \leq \bar{\gamma} \leq \bar{\gamma}_u, \quad \beta_l \leq \beta \leq \beta_u, \label{eq:GP_SIR2_bound2} \\
& \bar{\gamma} \in \Rb^n_{+}, w \in \Rb^n_{+}, \beta \in \Rb^{|E|}_{+}.
\end{align}
\end{subequations}
The above formulation corresponds to imposing NPIs in a cost-optimal manner while ensuring that the reproduction number stays below a certain threshold. 

\begin{remark}
Note that in Section \ref{section:least-square} we estimated the virus spread parameters while in this section we allow a social planner to control the spread of the virus by setting these parameters. There are two ways to interpret it: 1) before the social planner is able or decides to exert efforts to mitigate the spread of the virus, it must estimate the state of the system from testing data which, given the delay $\tau$, requires estimating the spread parameters, and 2) even after the social planner implements preventative measures, the population may not follow the restrictions or they may be less effective (or more extreme) than needed; therefore repeatedly estimating the actual parameters is necessary.
\end{remark}

We now evaluate the performance of the proposed framework via simulations and analysis of real data.
\section{Empirical Evaluation: Synthetic Data}\label{section:numerical}

\subsection{Inference from testing data}

\rev{We first illustrate the effectiveness of the proposed inference and parameter identification methods. Recall that the hyper-parameter $\alpha$, defined as the ratio of the probability of an infected individual undergoing testing and the probability of a healthy individual undergoing testing, plays a key role in our inference scheme. We now show how it can be learned from testing data following the procedure discussed in Section \ref{sec:alpha_tuning}.}

\begin{figure}
\centering
    \includegraphics[scale=0.6]{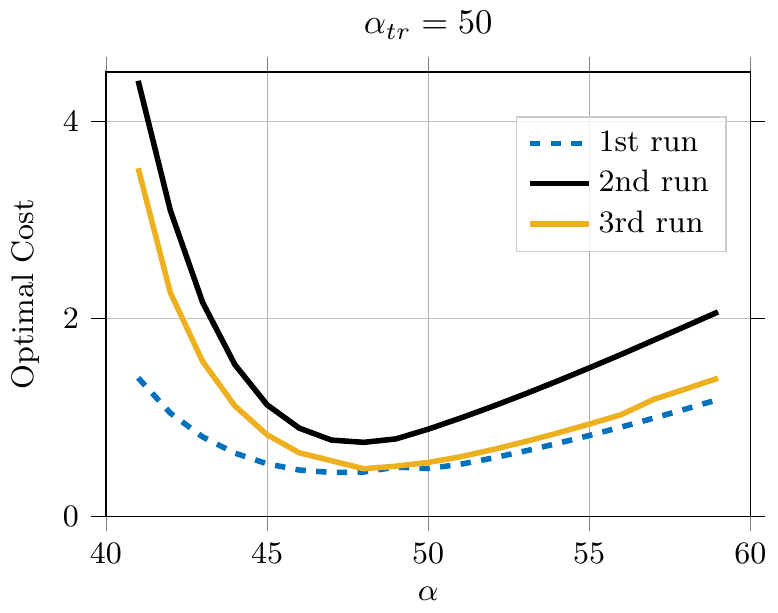}    
    \caption{Variation of optimal least squares cost \eqref{eq:ls_opt_new} for different values of $\alpha$ for $n = 5$ and $\alpha_{tr} = 50$.}
    \label{fig:opt_alpha}
\end{figure}

\rev{We generate random network topologies on $n$ nodes where there is an edge between any two nodes with probability $0.25$. For each such network, we choose infection rates to be uniformly distributed between $0.03$ and $0.05$ and the recovery rate for each node to be uniformly distributed between $0.01$ to $0.03$. We let the SIR epidemic state trajectory evolve following \eqref{eq:sir_si} with an initial infected proportion of $0.01$ at a few randomly chosen nodes. We then generate synthetic testing data according to our observer model defined in \eqref{eq:cik_bin} and \eqref{eq:nonlin_obs_d} with the number of daily tests being uniformly distributed between $2000$ and $2050$ at each node and the (true) value of $\alpha_{tr}$ being $10$, $50$, and $100$.}

\begin{figure*}
\centering
    \includegraphics[scale=0.6]{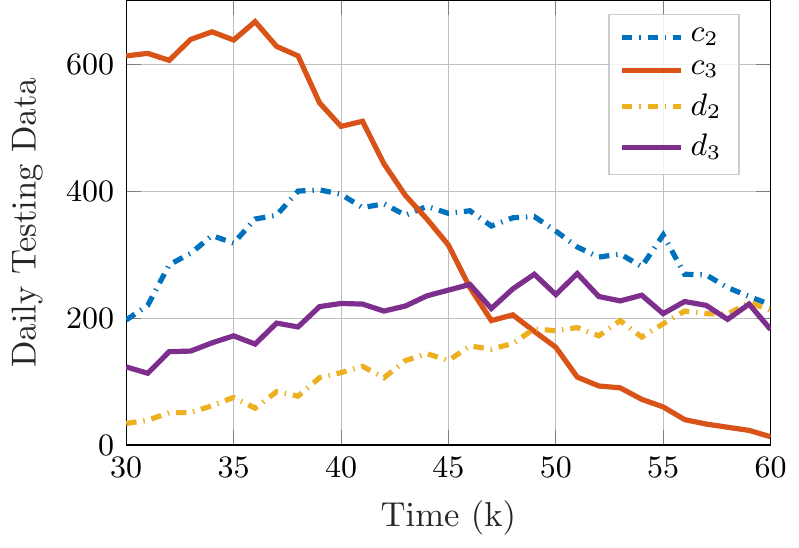}
    ~\includegraphics[scale=0.6]{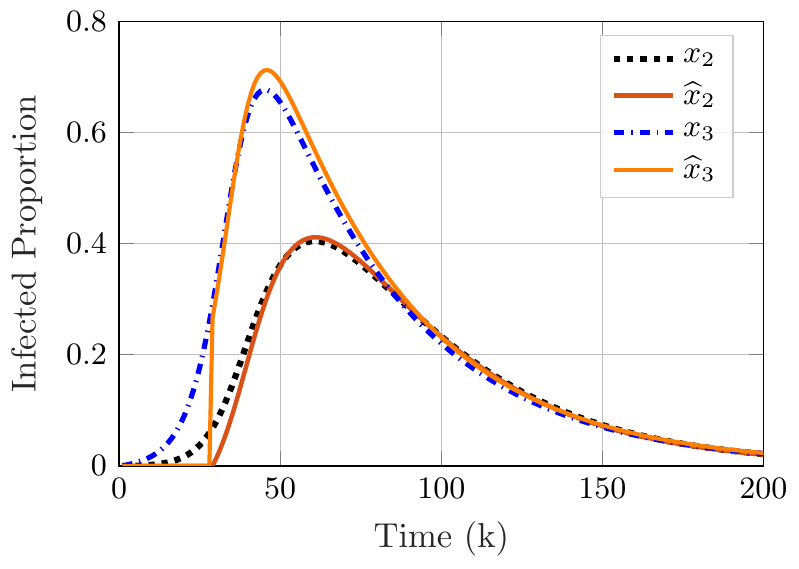}
    ~\includegraphics[scale=0.6]{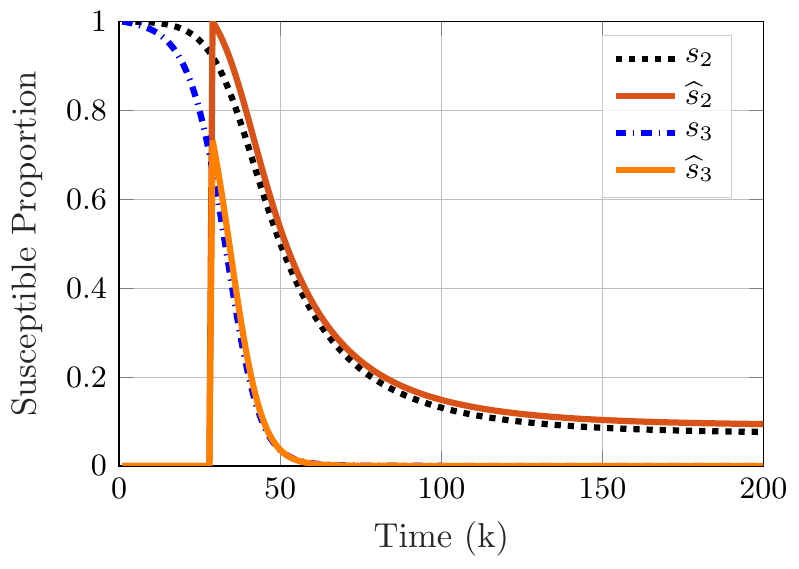}
    \caption{\small Synthetic testing data and comparison of true and predicted infected and susceptible proportions at two different nodes.}
    \label{fig:learned_vs_predicted}
\end{figure*}

\rev{We now use the synthetic testing data generated above to learn the (true) value of $\alpha_{tr}$ as well as the parameters of the epidemic model. We use testing data from $T_1 = 30$ and $T_2 = 60$, i.e., when the testing data is available once the disease has spread significantly among the population. We solve the (non-convex) least squares problem formulated in \eqref{eq:ls_opt_new} using optimization package YALMIP \cite{lofberg2004yalmip} with `fmincon' solver in MATLAB environment on a Desktop computer with an Intel Xeon 3.4GHz processor and 32GB of RAM for different integer values of $\alpha$ over the range $0.5\alpha_{tr}$ to $2\alpha_{tr}$ and examine the optimal cost value.} 

\rev{In Table \ref{table:learn_alpha}, we report the mean and standard deviation of the values of $\alpha$ that achieve the smallest optimal cost for different values of $\alpha_{tr}$ over $10$ independent runs; each run with a different randomly chosen network topology and epidemic parameters. We consider networks with $n = 5$ and $n=10$.\footnote{\rev{We note that computation time was significantly higher (more than one hour) for values of $n \geq 15$ due to the non-convexity of the problem.}} The result shows that the optimal least square cost is smallest when $\alpha$ is approximately close to $\alpha_{tr}$ with a relatively low variance. In Fig. \ref{fig:opt_alpha}, we show how the optimal cost varies as a function of $\alpha$ for three different instances with $n = 5$ and $\alpha_{tr} = 50$. The variation is qualitatively similar in all other instances we have considered.}

\begin{table}[htb]
\centering
\caption{\small \rev{Learned value of $\alpha$ for different network sizes and $\alpha_{tr}$ for $10$ independent runs.}}
\label{tab:case3}
\begin{tabular}{l c  c  c c}
{} & {} & $\alpha_{tr} = 10$ & $\alpha_{tr}  = 50$ & $\alpha_{tr}  = 100$\\
\toprule
& Mean & $10.25$& $47$& $95.7$\\
$n = 5$ & Std. & $1.03$ &  $1.41$ & $2.87$\\ 
& Worst Dev. &$12$ & $45$ & $91$\\
\midrule
& Mean & $9.11$& $46.8$& $94.6$\\
$n = 10$ & Std. & $0.78$&  $2.25$& $3.17$\\ 
& Worst Dev. & $8$ &$44$ & 87 \\
\bottomrule
\end{tabular}
\label{table:learn_alpha}
\end{table}

\rev{We now illustrate how the inferred or predicted epidemic states compare with the true states. We consider a randomly generated network on $10$ nodes. The plot in the left panel of Fig. \ref{fig:learned_vs_predicted} shows the synthetic testing data at two nodes between $T_1 = 30$ to $T_2 = 60$ with $\alpha_{tr} = 10$. For ease of exposition, we assume $\tau = 0$. We first solve the problem in \eqref{eq:ls_opt_new} for different values of $\alpha$ and learn that the optimal least squares cost is smallest for $\alpha = 9$. We then learn the infection and recovery rates as well as the proportion of susceptible and infected nodes at time $T_1-1$ by solving \eqref{eq:ls_opt_new}. Finally, we predict the SIR epidemic states using the learned values following \eqref{eq:for}.}

\rev{The evolution of the true and predicted proportions of infected and susceptible individuals for the two different nodes ($v_2$ and $v_3$) are shown in the middle and right panels of Fig. \ref{fig:learned_vs_predicted}. As the figure shows, for node $3$, the epidemic has spread significantly by time $T_1$. However, our estimate of the initial proportion of infected and susceptible individuals lead to a very accurate predicted state trajectory. We also note that while our estimates of initial states and recovery rates tend to be quite accurate, the learned values of infection rates ($\beta_{ij}$ values) tend to be very different from their true values; nevertheless, the predicted states closely align with the true states, consistent with the observations in \cite{prasse2020network} (also see Remark \ref{remark:prasse}).}


\subsection{Performance of optimal NPIs}

\begin{figure}
\begin{subfigure}{.45\linewidth}
  \centering
  \begin{tikzpicture}[-,>=stealth',shorten >=1pt,auto,thick,main node/.style={circle,draw,minimum size = 0.6cm,inner sep=0pt}]
  
  \node[main node] (1) {$\mathtt{CH}$};
  \node[main node] (2) [below = 0.7cm of 1] {$\mathtt{IT}$};
  \node[main node] (3) [above = 0.7cm of 1] {$\mathtt{DE}$};
  \node[main node] (4) [right = 0.8cm of 1] {$\mathtt{AT}$};
  \node[main node] (5) [left = 0.8cm of 1] {$\mathtt{FR}$};

  \path[every node/.style={font=\sffamily\small}]
    (1) edge (3)
    (1) edge (5) 
    (1) edge (4)     
    (1) edge (2)
    (3) edge (5)
    (2) edge (5)
    (2) edge (4)
    (3) edge (4);
\end{tikzpicture}  
\label{fig:graph}
\end{subfigure}
\begin{subfigure}{.45\linewidth}
  \centering
  \includegraphics[width=\linewidth]{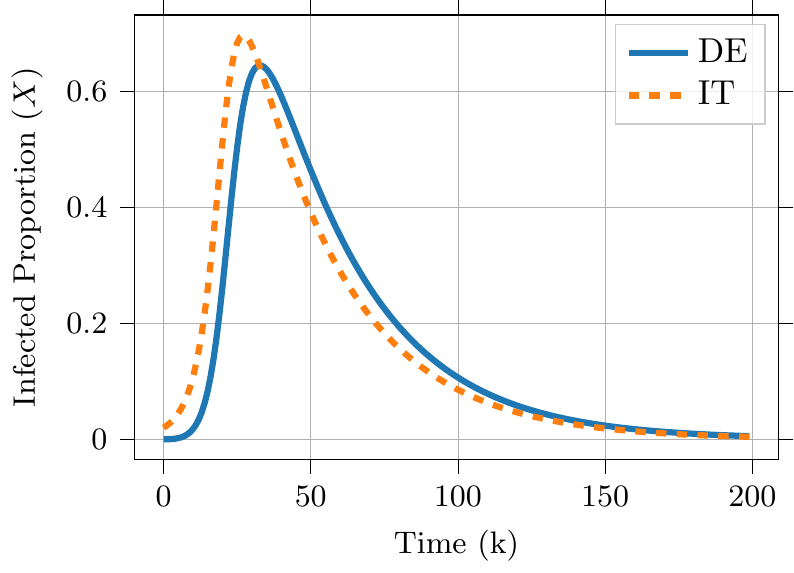}
    \label{fig:dynamics}
\end{subfigure}
\caption{Network topology and evolution of proportion of infected subpopulations without NPIs at nodes DE and IT.}
\label{fig:baseline}
\end{figure}

We now evaluate the effectiveness of the proposed online NPI allocation scheme via simulations. We consider a network with $5$ nodes whose topology is given in Fig.~\ref{fig:baseline}. The nodes are roughly modeled after five countries in continental Europe: France (FR), Germany (DE), Italy (IT), Austria (AT) and Switzerland (CH). Two nodes are neighbors if they share a border. \rev{The chosen infection and recovery rates are shown in Table \ref{tab:beta_gamma_baseline}, which satisfy Assumption \ref{assm:1} with $h = 1$.} The evolution of the infected proportions under these parameters at two nodes (DE and IT) with the initially infected proportions being $0.02$ at node IT and $0$ elsewhere is shown in Fig. \ref{fig:baseline}.

\begin{table}[htb]
\centering
\caption{\small \rev{Baseline infection and recovery rate parameters}}
\begin{tabular}{lllllll}
 & \multicolumn{5}{c}{$\beta$}         & $\gamma$ \\ \cline{2-6}
                  & DE   & FR   & AT   & IT   & CH   &                        \\ \hline
DE                & 0.05 & 0.05 & 0.05 & 0    & 0.05 & 0.03                   \\ 
FR                & 0.05 & 0.2  & 0    & 0.03 & 0.05 & 0.03                   \\ 
AT                & 0.05 & 0    & 0.2  & 0.05 & 0.04 & 0.03                   \\ 
IT                & 0    & 0.03 & 0.05 & 0.2  & 0.05 & 0.03                   \\ 
CH                & 0.05 & 0.05 & 0.04 & 0.05 & 0.2  & 0.03                   \\ \hline
\end{tabular}
\label{tab:beta_gamma_baseline}
\end{table}

We solve the problems stated in \eqref{eq:GP_SIR} and \eqref{eq:GP_SIR2} \rev{using the Python solver GPkit \cite{gpkit}}, and evaluate their performance in controlling the spread of the SIR epidemic. Following \cite{preciado2014optimal}, we consider the following cost functions for NPIs: 
\begin{equation}\label{eq:NPI_cost_fn}
f_{ij}(\beta_{ij}) = \frac{\beta_{ij}^{-1} - \beta^{-1}_{u, ij}}{\beta^{-1}_{l, ij} - \beta^{-1}_{u, ij}}, \quad \text{and} \quad g_{i}(\bar{\gamma}_{i}) =  \frac{\bar{\gamma}_{i}^{-1} - \bar{\gamma}_{u,i}^{-1}}{\bar{\gamma}_{l,i}^{-1} - \bar{\gamma}_{u,i}^{-1}};
\end{equation}
recall that $\bar{\gamma}_{i} = 1-h\gamma_{i}$. 

\begin{figure*}
\centering
    \includegraphics[scale=0.6]{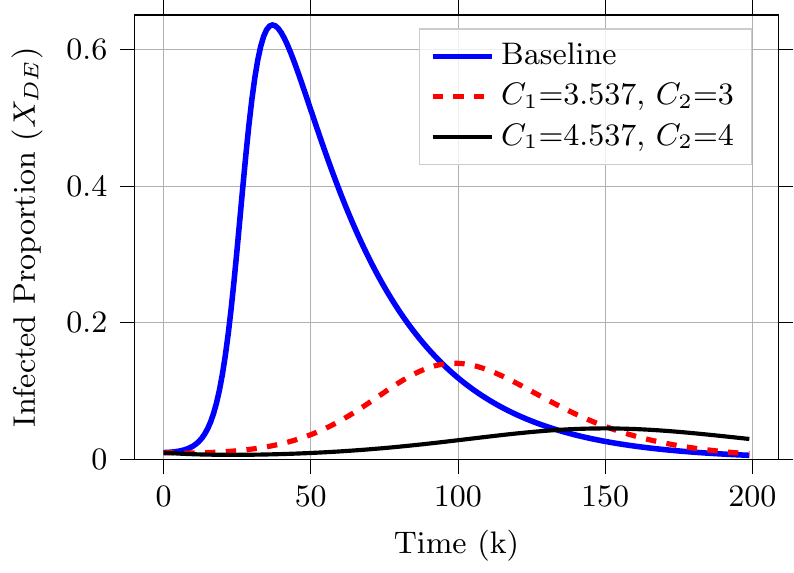}	  
    ~\includegraphics[scale=0.6]{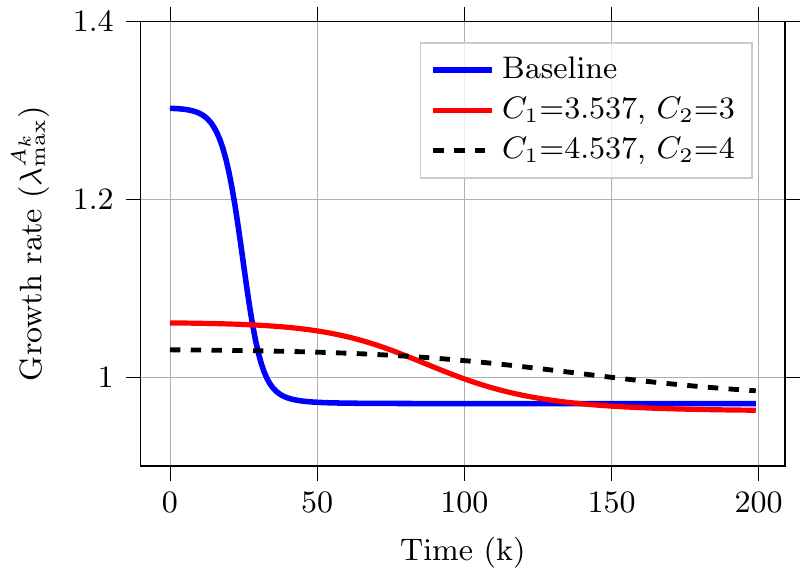}
    ~\includegraphics[scale=0.6]{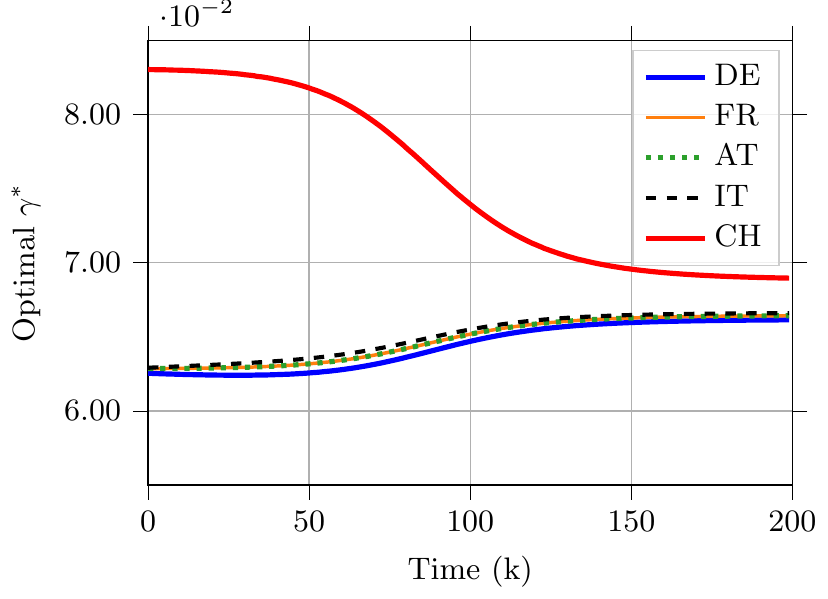}
    \caption{\small Evolution of the infected proportions at node DE, optimal growth rate ($\lambda^*_k$) and optimal recovery rates ($\gamma^*$) under optimal budget-constrained NPIs obtained by solving \eqref{eq:GP_SIR}. Baseline scenario refers to the case without NPIs with parameters shown in Table \ref{tab:beta_gamma_baseline}.}
    \label{fig:curve_flattening}
\end{figure*}

\begin{figure*}
\centering
    \includegraphics[scale=0.6]{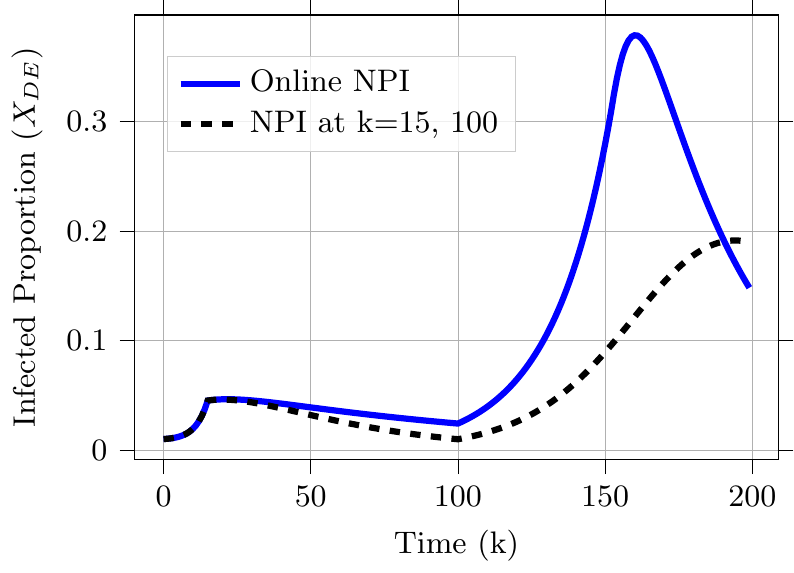}	  
    ~\includegraphics[scale=0.6]{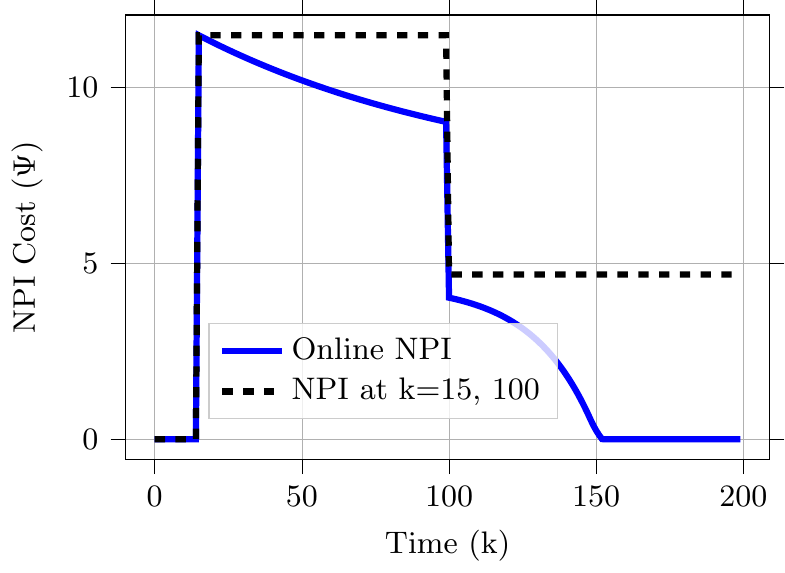}
    ~\includegraphics[scale=0.6]{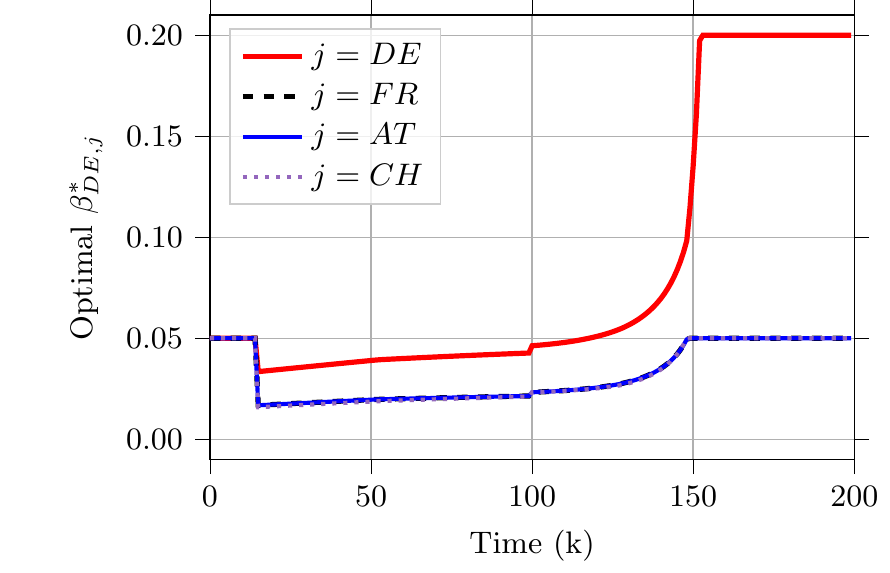}
    \caption{\small Evolution of the infected proportion and the optimal cost of NPIs obtained by solving \eqref{eq:GP_SIR2} for different constraints on $\lambda^{A_k}_{\max}$.}
    \label{fig:secondwave}
\end{figure*}

Note that $f_{ij}$ is monotonically decreasing (i.e., it is costly to reduce contact rates), \rev{with $f_{ij}(\beta_{l,ij}) = 1$, $f_{ij}(\beta_{u,ij}) = 0$, and $f_{ij}(\beta_{ij}) \in [0,1]$ for $\beta_{ij} \in [\beta_{l, ij},\beta_{u, ij}]$. In other words, as $\beta_{ij}$ decreases from $\beta_{u,ij}$ to $\beta_{l,ij}$, the corresponding cost increases from $0$ to $1$.} The function $g_i$ has analogous properties. 

\rev{In our evaluation, we choose the bounds such that $\bar{\gamma}_i \in [0.91,0.97]$, $\beta_{ii} \in [0.02,0.2]$ for $i \in [n]$ and $\beta_{ij} \in [0.005,0.05]$ when $i \neq j$.} The upper and lower limits on $\beta_{ii}$ and $\beta_{ij}$ are chosen differently to reflect the fact it is easier to reduce contact between individuals from different subpopulations compared to individuals within a subpopulation.

In order to isolate the performance of the online optimization approaches described in Section \ref{section:optimization}, we assume that $s(k)$ is exactly known, and report the results under NPIs obtained by solving \eqref{eq:GP_SIR} at every instance for different budget combinations. We assume that the initial fraction of infected nodes is $0.01$ in DE, and $0$ for all other nodes. 

Fig.~\ref{fig:curve_flattening} shows the evolution of the proportions of the infected individuals at node DE, and the instantaneous growth rate $\lambda^{A_k}_{\max}$ at the optimal NPIs. The figure shows that higher budgets lead to ``flattening" of the curve of the infected proportion. The peak of the infected proportion is smaller and occurs much later compared to the baseline setting without NPIs. In addition, we observe that the cumulative number of people who become infected shows significant reduction depending on the budget. When the budget is sufficiently high, not many people become infected in the first place. Consequently, the proportion of susceptible individuals remains high, which (perhaps counter-intuitively) results in a larger value of $\lambda^*_k$ compared to the cases with smaller budgets for NPIs. Nevertheless, the infected proportion eventually decreases to $0$. \rev{The evolution of the optimal recovery rates with time for different nodes when $C_1 = 3.537$ and $C_2 = 3$ is shown in the right panel of Fig.~\ref{fig:curve_flattening}. Node CH receives a significantly higher share of the budget-constrained recovery rates, potentially due to its central location. These results highlight the importance of optimal allocation of NPIs exploiting the structure of the network.}

We then focus on computing cost optimal NPIs subject to constraints on the growth rate or the reproduction number. We solve the optimization problem in \eqref{eq:GP_SIR2} starting from $T_1 = 15$ with constraint $\lambda_k = 0.99$ for $k\in[16,100]$ and $\lambda_k = 1.05$ for $k\in[101,200]$. The evolution of the infected proportion at node DE and the optimal cost obtained by solving \eqref{eq:GP_SIR2} at every time instance are shown as solid (blue) curves in Fig.~\ref{fig:secondwave}. The dashed (black) curves show the evolution when the NPIs are computed only at $k=15$ and at $k=100$ and maintained over the intervals $k \in [16,100]$ and $k \in [101,200]$. \rev{The evolution of the infection rate parameters $\beta_{DE,j}$ for $j \in \{DE, FR, AT, CH\}$ are shown in the right panel of Fig.~\ref{fig:secondwave}. The figure shows that around $k = 150$, the proportion of susceptible individuals become negligible leading to the growth rate being inherently smaller than the desired value, and as a result, the NPIs are gradually withdrawn.}

Fig.~\ref{fig:secondwave} shows that early interventions reduced the growth of the epidemic, albeit at a high cost of NPIs. As a result, the proportion of susceptible individuals remains high for $k \in [16,100]$. Consequently, the baseline reproduction number remains high and it remains costly to maintain the growth-rate below $1$. Furthermore, if NPIs are withdrawn before the infected proportion becomes negligible, we observe an exponential increase in the infected proportion (i.e., a ``second wave" of infections). 

Fig.~\ref{fig:secondwave} also highlights the importance of regularly adapting the NPI strategy using feedback from testing data. Recall from Theorem \ref{prop:main} that the baseline reproduction number is monotonically decreasing. Accordingly, if optimal NPIs are computed at every time step, then the optimal cost shows a gradual decrease and the optimal solution pertains to a gradual easing of NPIs while maintaining the desired growth rate. In contrast, when updated sporadically (as is the case with the dashed black curve), the NPIs are more stringent. This achieves better performance in reducing the spread of the epidemic, but at a much higher cost.
\section{Empirical Evaluation: Real Data}\label{section:real_empirical}

We have carried out an exploratory analysis of publicly available data for COVID-19 from \cite{Data_COVID} for the five countries mentioned above under the proposed framework. We have used testing data for a period of \rev{eleven months from 1st March 2020 to 31st January 2021}. The testing data is comprised of the number of tests carried out ($z(k)$), the number of confirmed positive cases ($c(k)$), and the number of recoveries ($d(k)$). Note that the death counts are included in the recovery data. Some dates have missing data on the number of tests being carried out. In such cases, linear interpolation of the past and next available data points are used \cite{ciss}. We have filtered the original data by taking a $7-$day moving average to reduce the effect of outliers. \rev{The plots in the left panel of Fig. \ref{fig:real_scalar_inf} show the test positivity rate ($c(k)/z(k)$) and active cases for two countries, Germany (DE) and Italy (IT), respectively.}

\rev{We emphasize that the results in this section are obtained by the proposed methodology which is based on the SIR epidemic dynamic and other assumptions stated in the paper. In the absence of large-scale serology tests and knowledge of testing policies, it is difficult to quantify the accuracy of the inference results in general, and the same caveat applies to the results presented in this paper as well.}

\begin{figure*}
\centering
    \includegraphics[scale=0.6]{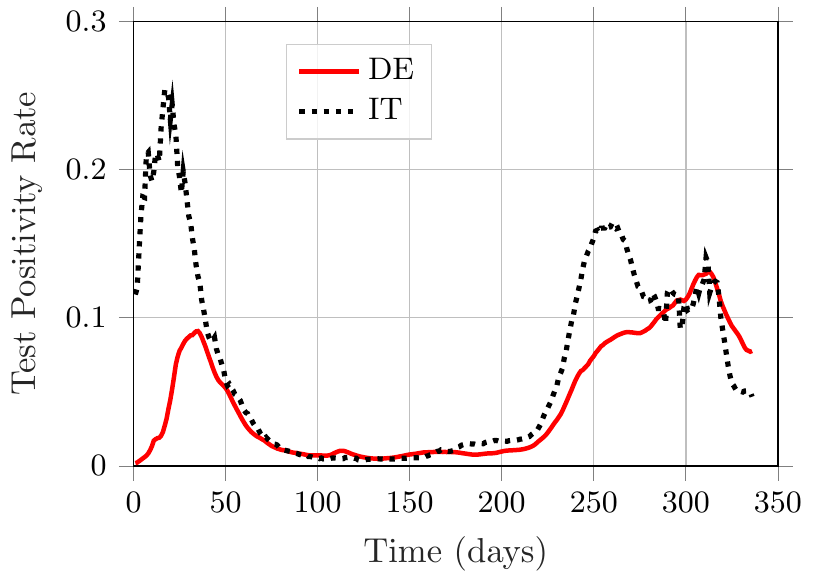}
    ~\includegraphics[scale=0.6]{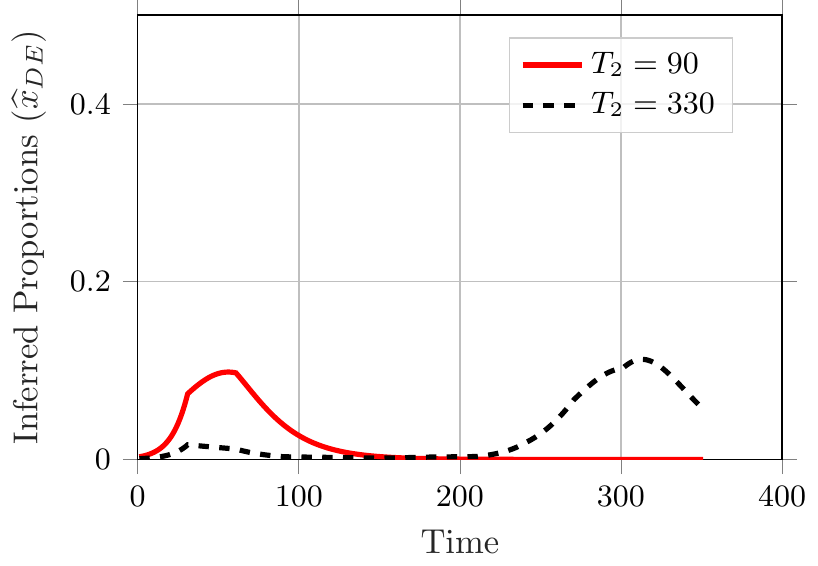}
    ~\includegraphics[scale=0.6]{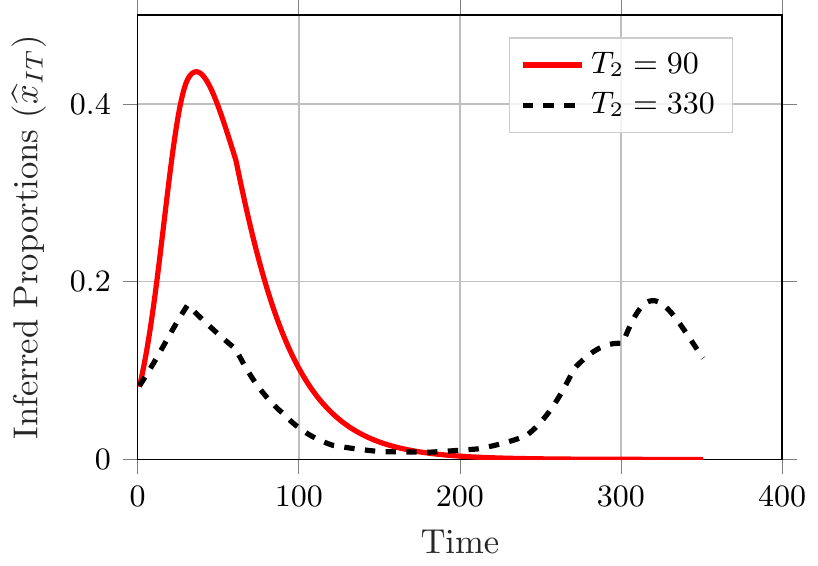}
    \\
    \includegraphics[scale=0.6]{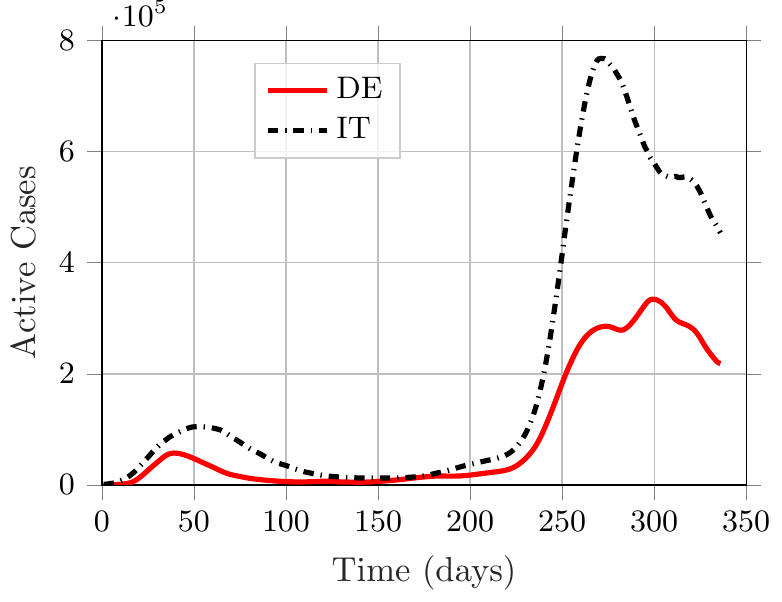}
    ~\includegraphics[scale=0.6]{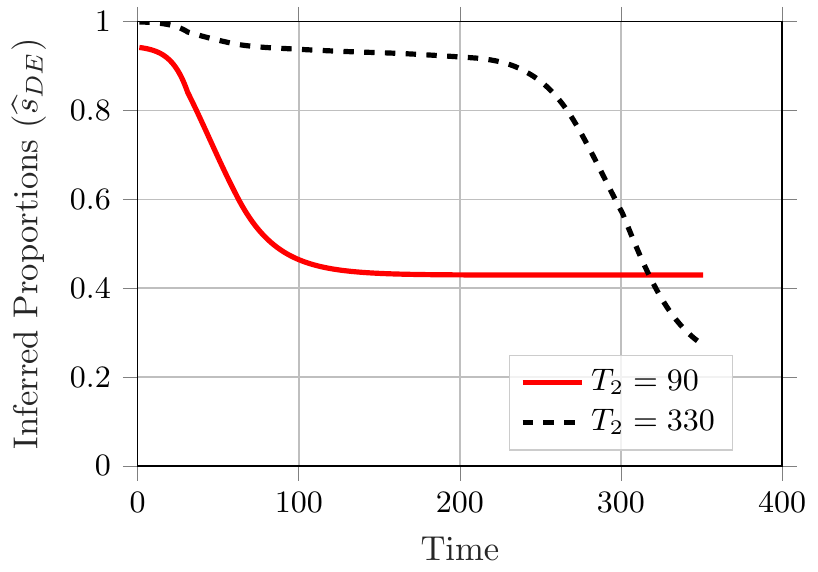}
    ~\includegraphics[scale=0.6]{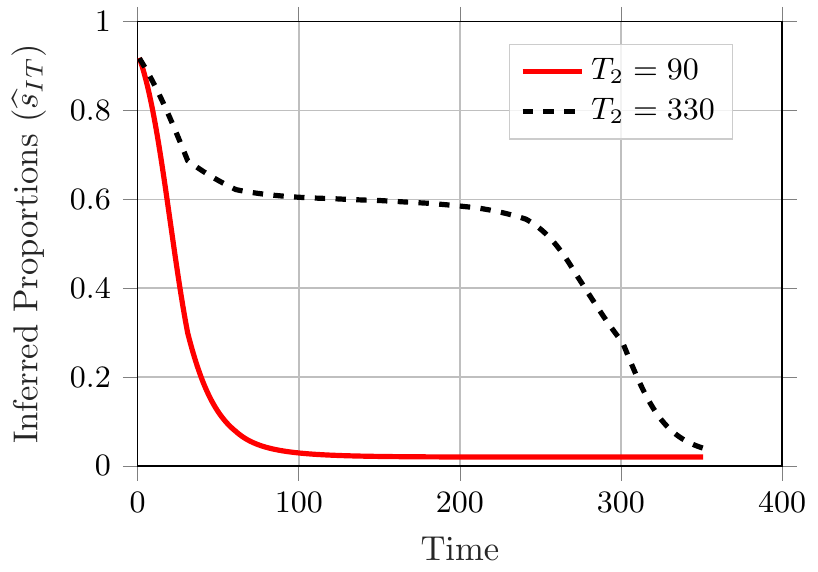}
    \caption{Real testing data and inferred/predicted states for Germany (DE) and Italy (IT) from 1st March 2020 to 31st January 2021 with values of $\alpha$ as stated in Table \ref{tab:inf_alpha_real} for the respective countries and values of $T_2$.}
    \label{fig:real_scalar_inf}
\end{figure*}

\subsection{Inferring epidemic states}

From the above data, we infer the underlying epidemic states, (i.e., the proportions of susceptible, infected, and recovered subpopulations in all five countries) \rev{and the SIR epidemic parameters following the methodology proposed in Section \ref{subsection_testing_risk} and Section \ref{section:least-squareb}.} \rev{We first analyze the testing data for each of the five countries separately with value of $\tau$ set to $0$. For each country, we consider the available testing data from $T_1 = 1$ to different values of $T_2 \in \{90,150,210,270,330\}$ to examine how the inferred states change as more and more testing data becomes available. For each value of $T_2$, we solve the least squares problem in \eqref{eq:ls_opt_new} assuming that $\beta$ and $\gamma$ parameters change every $30$ days. We then learn the initial inferred states and these parameters for integer values of $\alpha$ ranging from $1$ to $100$. The values of $\alpha$ that achieve the smallest least squares cost are reported in Table \ref{tab:inf_alpha_real}.}

\begin{table}[htb]
\centering
\caption{\small \rev{Values of $\alpha$ that achieve the smallest least squares cost for real testing data from $T_1 = 1$ and different values of $T_2$.}}
\begin{tabular}{lclclclclc|c|}
\toprule
$T_2$ & DE  & FR & AT & IT & CH                                       
\\ \midrule
$90$ & $6$  & $12$  & $6$ & 12 & $10$     \\
150   & 6      & 14    & 6    & 14 & 10      \\
210   & 8      & 19     & 10  & 15   & 12     \\
270   & 12    & 26  & 24  & 26 & 36      \\
330   & 24    & 30    & 24  & 31 & 40                             
\\ \bottomrule
\end{tabular}
\label{tab:inf_alpha_real}
\end{table}

\rev{The values in Table \ref{tab:inf_alpha_real} show that as more testing data becomes available, the value of $\alpha$ that best explains the testing data increases. This is evident from the inferred and predicted quantities shown in the middle and right panel of Fig. \ref{fig:real_scalar_inf} for Germany and Italy, respectively. The solid (red) curve shows the predicted infected and susceptible proportions obtained via solving \eqref{eq:for} with the learned values of initial states and parameters when testing data from $T_1 = 1$ to $T_2 = 90$ is used with the value of $\alpha$ shown in the first row of Table~\ref{tab:inf_alpha_real}.} \rev{Fig. \ref{fig:real_scalar_inf} shows that a significant proportion of individuals became infected during the first wave of COVID-19. In fact, the peak infected proportion was more than $0.4$ for Italy.}

\begin{figure*}[tb]
	\centering
	\includegraphics[scale=0.6]{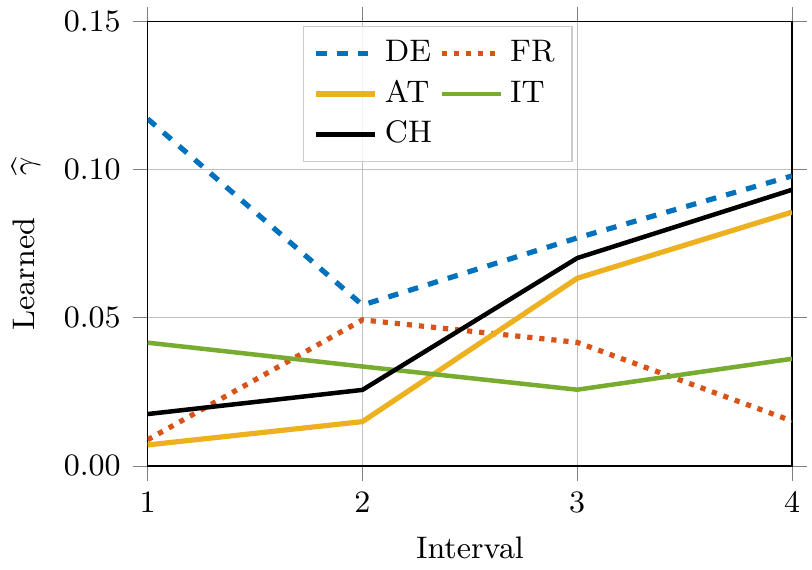}~
	\includegraphics[scale=0.6]{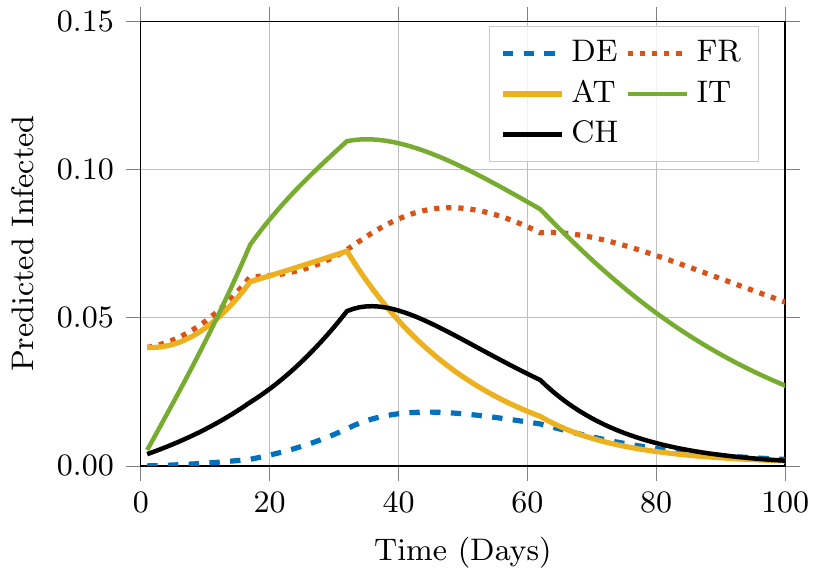}~
	\includegraphics[scale=0.6]{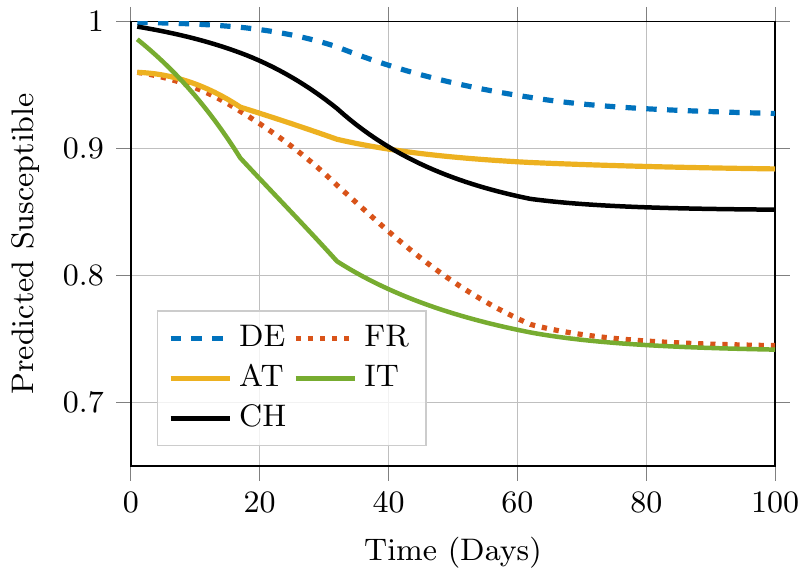}~
	\\
	\includegraphics[scale=0.6]{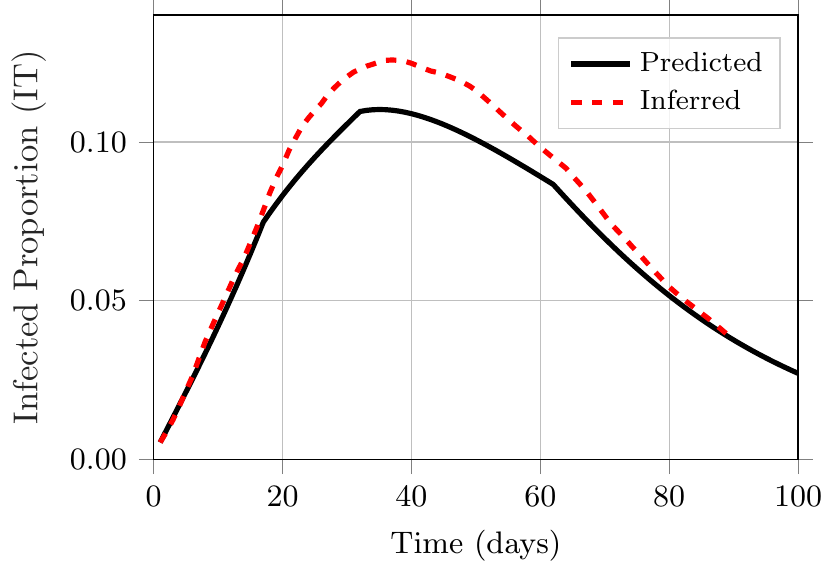}~
	\includegraphics[scale=0.6]{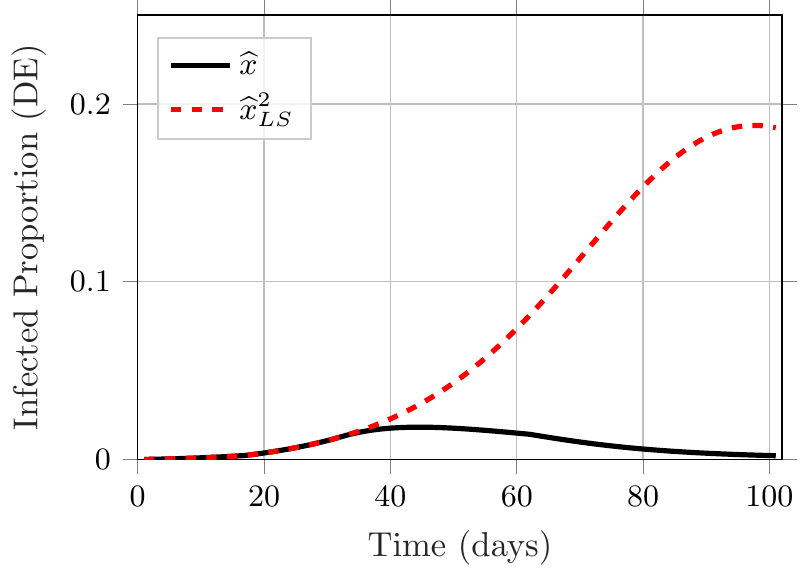}~
	\includegraphics[scale=0.6]{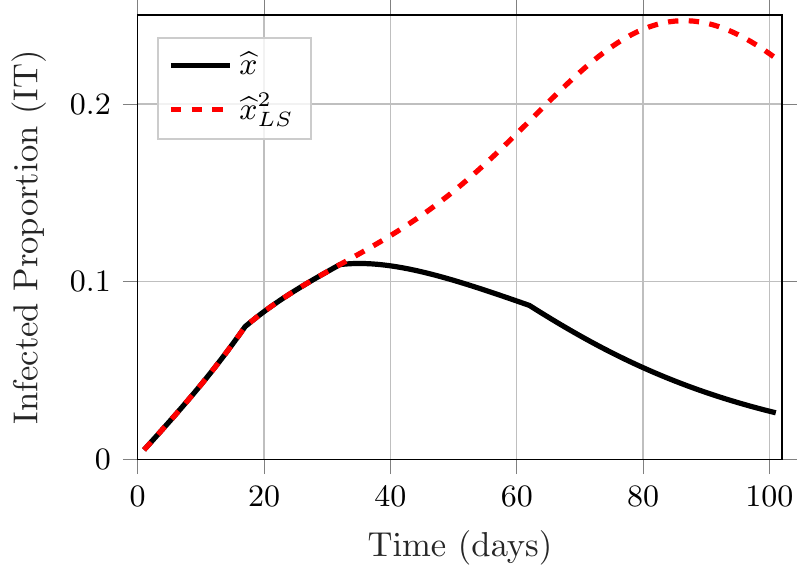}~
	\caption{\footnotesize{The top left panel shows the learned values of recovery/curing rates for different sub-intervals. The middle and right panels in the top row show the predicted proportions of infected and susceptible individuals. The bottom left panel compares the inferred and predicted infected proportions for Italy. The middle and right panels in the bottom row show the predicted and projected infections without NPIs for Germany and Italy; $\widehat{x}^2_{\mathtt{LS}}$ is the forecasted trajectory if the parameters were kept unchanged after the second sub-interval.}}
    \label{fig:IT_ls}
\end{figure*}

\rev{However, data from the second wave of infections indicates that the fraction infected during the first wave was in fact significantly smaller. As a result, the previous value of $\alpha$ (learned using testing data till $T_2 = 90$) results in an infeasible trajectory under testing data from $T_1 = 1$ to $T_2 = 330$. The dashed (black) line in Fig. \ref{fig:real_scalar_inf} shows the inferred/predicted quantities with value of $\alpha$ shown in the bottom row of Table~\ref{tab:inf_alpha_real}. The figure also indicates that the learned initial proportion of infected nodes was around $0.1$ for Italy at the beginning of March while it was negligible for Germany.}

\rev{Thus, our analysis provides valuable insights into the effect of the tuning parameter $\alpha$, the significance of learning the initial infected proportion, and how to tune the inference technique as more and more testing data becomes available. Nevertheless, we clarify that while the values of $\alpha$ reported in Table \ref{tab:inf_alpha_real} should be viewed as quantities that best explain the data assuming that $\alpha$ remains unchanged over the entire time period. In practice, this parameter, which captures the ratio of the likelihood of an infected individual undergoing testing and a healthy individual undergoing testing, should be chosen/learned in collaboration with public health officials. This parameter is also likely time-varying as testing practices evolve over the course of the pandemic. The inference can be further tuned using outcomes of large-scale serological surveys by imposing suitable lower bounds on the proportion of susceptible individuals at different points of time. Further investigations along these lines are beyond the scope of this paper and will be explored in a follow up work.}

\subsection{Inference and estimation in the networked model}

\rev{We now learn the parameters of the networked SIR epidemic dynamic, with the network structure as shown in Fig. \ref{fig:baseline}, by solving the least squares problem in \eqref{eq:ls_opt_new} for the first wave data of $90$ days from 1st March 2020 to 29th May 2020. We divide the time-span into four sub-intervals from day $1$ to $15$, $16$ to $30$, $31$ to $60$ and $61$ to $90$ to better capture the changing lockdown measures being put in place during those periods. Specifically, we learn four different sets of $\beta_{ij}$ and $\gamma_i$ parameters where each corresponds to different sub-intervals. We choose $\alpha = 40$ for our inference as it would be consistent with both first and second wave data as discussed earlier. We also impose the constraint that the initial infected proportion to not exceed $0.05$ at any of the nodes.}

\rev{The learned values of the curing or recovery rates $\gamma$ for different countries and sub-intervals are shown in the top left panel of Fig. \ref{fig:IT_ls}. The learned parameters show that during the initial period, the recovery rates were small which contributed to a rapid increase in the number of infections. However, the recovery rates showed a gradual increase over time except in France and Italy where the variation was not significant. The predicted infected and susceptible proportions (by solving \eqref{eq:for} starting from the initial inferred states and the learned $\beta$ and $\gamma$ parameters) are shown in the middle and right panels in the top row of Fig. \ref{fig:IT_ls}.} The inferred states with delay $\tau=8$, are found to be very similar to a shifted version of the results obtained with $\tau=0$. These results are shown in the supplementary appendix.

\rev{In the bottom left panel of Fig. \ref{fig:IT_ls}, we compare the inferred infected proportion defined via the learned initial proportions and the testing data following \eqref{eq:inf_x_final}, and the predicted infected proportion determined following \eqref{eq:for} starting from the learned initial proportions and using the learned $\widehat{\beta}$ and $\widehat{\gamma}$ parameters for the different sub-intervals for Italy. Both quantities are largely in agreement showing that the learned $\beta$ and $\gamma$ parameters well explain the inferred states. However, note that this plot does not indicate the accuracy of the inferred states since the ground truth is not known. The learned $\beta$ parameters are shown in the supplementary appendix.}

\subsection{Learned vs. optimal NPIs}

\begin{figure*}[tb]
	\centering
	\includegraphics[scale=0.6]{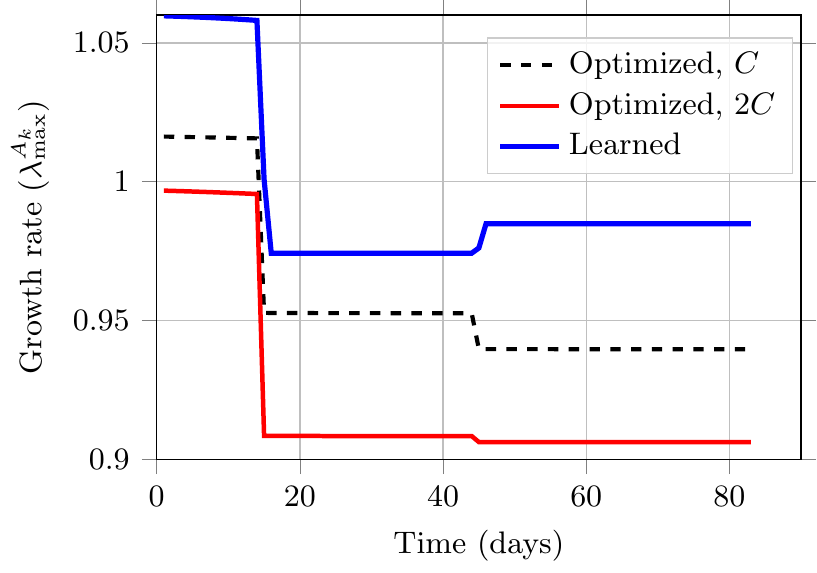}~
	\includegraphics[scale=0.6]{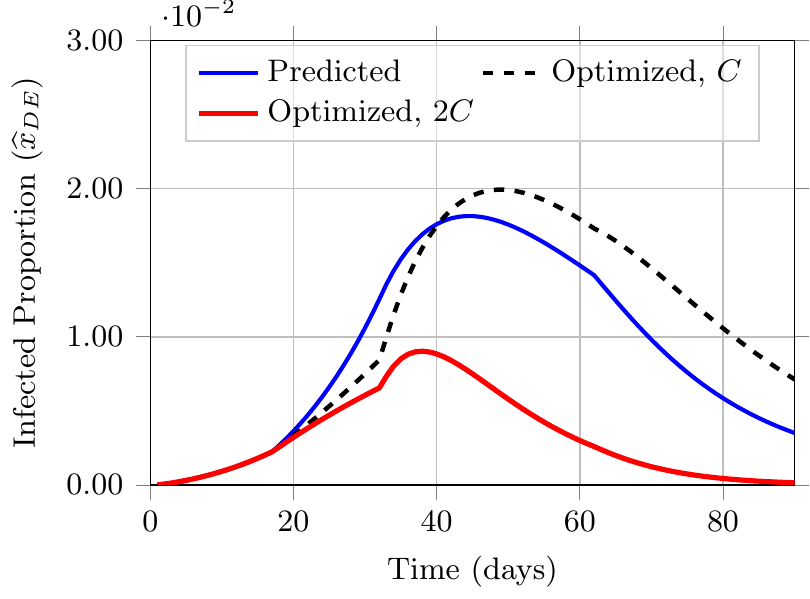}~
	\includegraphics[scale=0.6]{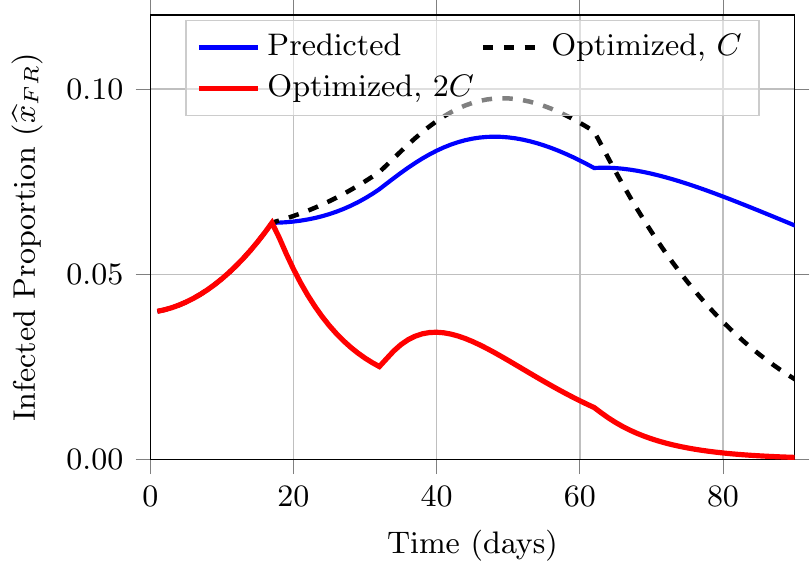}
	\caption{Predicted and optimized trajectories of the growth rate and the infected proportions at nodes DE and FR under optimal recovery rates computed by solving \eqref{eq:GP_SIR} with the original cost (C) and twice the original cost (2C) as budgets. Time/day $0$ on the left panel corresponds to time/day $17$ in the middle and right panels.}
    \label{fig:NPI_Real1}
\end{figure*}

We now illustrate the impact of learned NPIs by computing the trajectory under the parameters learned for the first and the second sub-intervals. Specifically, let $\widehat{x}^2_{\mathtt{LS}}$ be the trajectory of the proportion of infected individuals if the infection and recovery rates were kept equal to the parameters learned from data corresponding to the second sub-interval (days $16$ to $30$). \rev{The plots in the middle and right panels in the bottom row of Fig. \ref{fig:IT_ls}} show that NPIs such as imposing lock-down measures and augmenting healthcare resources during the second half of March significantly reduced the infected proportion compared to the initial growth rate, \rev{which if unchecked, would have led to a significantly higher peak.}

We further explore the optimality of the recovery rate parameters learned from the observed data. \rev{We determine the total cost (C) of the learned recovery rates under the cost function $g_i$ from \eqref{eq:NPI_cost_fn} for different sub-intervals to be $1.388, 1.194, 2.076$ and $2.567$, respectively. In the cost function, we set the upper and lower bounds to be the largest and smallest among the learned parameters for each country over all sub-intervals.\footnote{\rev{Note that the cost function is inversely proportional to $\bar{\gamma}$ and the bounds play the role of a normalization constant. If we allow the upper and lower bounds to be more spread out than the learned parameters, then the denominator in $g_i$ will increase and as a result, performance of the optimized trajectory would be ``better" due to more relaxed constraints and a smaller cost value. Thus, our choice of bounds allows for a fair comparison with the learned parameters.}}} 

We now compute the optimal recovery rates to minimize the growth rate of the epidemic subject to the budget constraints on the cost of the recovery rates. Specifically, we solve \eqref{eq:GP_SIR} for the optimal $\gamma$ once for each sub-interval starting from day $16$. We do not treat the infection rates as decision variables, and use their learned values for the respective sub-intervals. We choose the budgets to be the total cost of the learned NPIs stated above (C) (that is, a reallocation of the existing resources), and twice the cost (2C). The resulting growth rates and the trajectory of the infected proportions for DE and FR are shown in Fig. \ref{fig:NPI_Real1}. The figure shows that with the same total cost of NPIs, the optimal recovery rate allocation would have significantly reduced the infected proportion for FR with a relatively minor increase in the infected proportion of DE. When the optimal NPIs are imposed under twice the budget, the performance is significantly better for all the nodes for all sub-intervals. Thus, our results in this section highlight the importance of a closed-loop approach, i.e., using the inferred states as feedback to compute optimal NPIs in order to mitigate and control networked epidemics such as COVID-19.

\section{Discussions and Extensions}\label{sec:conc}

\rev{In this section, we briefly describe the generality and possible extensions of the proposed framework and highlight promising directions for future research.}

\subsection{Extensions and generalizations}

\subsubsection{Presence of asymptomatic infections} 

\rev{In this work, we focus on the SIR epidemic dynamic as it is one of the most fundamental mathematical models of epidemic evolution. In order to capture the presence of asymptomatic carriers in COVID-19, and the fact that early testing was mostly carried out for (high risk) symptomatic individuals, we introduce the delay factor $\tau$ in our inference strategy. Related epidemic models have explicitly captured asymptomatic infections as an additional epidemic state \cite{ansumali2020modelling}. In order to infer the fraction of nodes in the asymptomatic state, one needs data regarding the number of tests carried out on asymptomatic individuals and the number of confirmed cases therein. Since publicly available testing data does not have this information, we did not explore this further. Nevertheless, analyzing such higher order models remains an interesting avenue for future research.}

\subsubsection{Splitting subpopulations via different criteria} 

\rev{The salient feature of the proposed NPI computation approach is that if the proportion of susceptible individuals in different subpopulations are known, then the evolution of the infected proportions follows a linear dynamics characterized by a nonnegative irreducible matrix. Consequently, the problem of minimizing the spectral radius of this matrix can be formulated as a GP. Therefore, the proposed approach is still valid when the individuals are divided into different subpopulations based on their degree of interaction, age, comorbidity status, among others \cite{grundel2020much}. The interpretation of the NPIs would then be in terms of rate of interaction among individuals across different subpopulations and rate of recovery for each subpopulation (for instance, a social planner could assign higher medical resources to the older population and allow increased interaction among younger or healthier subpopulations.)}

\rev{Nevertheless, the decision-maker needs to know the the current susceptible proportion in each subpopulation which requires information about the number of tests carried out and the number of confirmed cases for different subpopulations. In other words, testing data needs to be made available at a higher granularity for efficient containment of epidemics.} 

\subsection{Avenues for future research}

\rev{We now summarize several interesting open problems for the community to explore moving forward.}

\subsubsection{Controlling generalized spreading processes} 

\rev{Compartmental epidemic models, such as the SIR and SIS epidemic models, have been used to study the spreading of computer viruses and malware \cite{dabarov2020heterogeneous}, opinions, rumours and memes \cite{dadlani2017mean, zhuang2016information,salehi2015spreading}, among others, and also on multi-layer networks \cite{sahneh2017contact,sahneh2014competitive}. In addition, competition and coexistence of multiple viruses \cite{yang2017bi,liu2019analysis} also introduce further complexity. Controlling the evolution of such processes in an online manner with limited observations is an important problem where suitable modifications of the proposed approach would be applicable.}

\subsubsection{Incorporating behavioral response to NPIs}

\rev{The success of NPIs depend on the effectiveness with which they are enforced. While the social planner computes the optimal contact rates leading to policy guidelines, individuals may or may not follow them based on their own perception of infection risk. Therefore, game-theoretic models of human decision-making under epidemic risks \cite{hota2019game,hota2020impacts} need to be integrated into the proposed closed-loop framework for improved prediction and control of spreading processes.}

\subsubsection{Robust and stochastic optimization}

\rev{Optimal control (including MPC based) approaches for epidemic containment require the current state information as feedback. However, the estimated state as well as certain model parameters that are used to compute optimal NPIs are uncertain and often differ when different forecasting techniques are deployed. Nevertheless, the decision-maker needs to ensure that the optimal solution performs well despite the uncertainty. Developing tractable robust and stochastic geometric programming formulations for optimal NPI computation remains an important avenue for future research.}

\subsubsection{Distributed inference and optimization} 

\rev{In this work, we solve the optimal control problems in a centralized fashion. This is motivated by the fact that in many instances, a central authority (such as the federal government of a country) decides to impose lockdown measures and other NPIs over a large geographical area with multiple regions. Nevertheless, the problem of solving the inference and optimal control problems in a distributed manner is of great interest. There have been a few recent approaches on distributed NPI computation for the SIS epidemic \cite{enyioha2015distributed,mai2018distributed,ramirez2018distributed} where a GP, analogous to our formulation, is solved in a distributed manner. One advantage of the proposed approach over the nonlinear MPC based approaches \cite{carli2020model,kohler2020robust} is that distributed convex optimization approaches can be readily applied to solve the GP. However, the accuracy of the estimated infection prevalence at each node of the network and whether nodes are willing to share accurate information (for example, due to political reasons) with neighboring nodes remain as important challenges towards developing distributed approaches.}

\subsubsection{Heterogeneous delay factors}

\rev{We introduce a constant delay factor to model the delay in onset of symptoms for infected individuals. However, the delay in onset of symptoms and due to testing is heterogeneous. A generalization to the case with heterogeneous delay parameters was studied in \cite{zhang2021estimation}. Epidemic state inference under heterogenous and distributed delays remains an interesting direction for future work.}

\section{Conclusion}
In this paper, we introduce a closed-loop framework to estimate and control the spread of the SIR epidemic on networks. We first rigorously establish the behavior of the discrete-time SIR epidemic dynamic; specifically that the dynamic is well-defined, the susceptible proportions and the growth rate decrease monotonically over time, the growth rate eventually falls below $1$, and after some point the infected proportions converge at least linearly to zero. We also incorporate several characteristics of real-world testing data in the state estimation task and examine the impacts of allocating NPIs (such as reducing contact rates and augmenting healthcare equipment and personnel) by solving suitable GPs in an online manner. Our results provide compelling insights on the behavior of the epidemic dynamics under NPIs and will be a valuable tool for policy-makers.

\section*{Acknowledgement}
The authors thank Prof. Shreyas Sundaram, Humphrey Leung, Baike She (Purdue University), and Damir Vrabac (Stanford University) for helpful discussions that contributed to this work. The authors thank Sanket Kumar Singh for his contributions in data analysis and visualizations.

\appendix
\section{Proof of Theorem \ref{prop:main}.}\label{section:proof}

\begin{proof}

We present the proof for each part of the theorem, starting with 1). 

1) We prove this result by induction. By assumption $s_i(0), x_i(0),$ $r_i(0) \in [0,1]$ and $s_i(0) + x_i(0) + r_i(0) = 1$ for all $i \in [n] $. Recall from Assumption \ref{assm:1} that $s_i(0)>0$ and $(1-h\sum^n_{j=1} \beta_{ij}) > 0$. Therefore, from \eqref{eq:sir_si}, we have $s_i(1) \geq (1-h\sum^n_{j=1} \beta_{ij}) s_i(0) > 0$. Since $ h \big[-s_i(0) \sum^n_{j=1} \beta_{ij} x_j(0) \big]\leq 0$, $s(1) \leq s(0)\leq 1$. By Assumption~\ref{assm:1} and \eqref{eq:sir_xi}, $x_i(1) \geq (1-h\gamma_i) x_i(0) \geq 0$. 
Since $x_j(0) \leq 1$ for all $ j\in [n]$, by \eqref{eq:sir_xi} and Assumption \ref{assm:1}, $x_i(1)\leq x_i(0) + s_i(0)h\sum^n_{j=1} \beta_{ij}\leq x_i(0) + s_i(0)\leq 1$. 
By \eqref{eq:sir_ri} and the non-negativity of $h$ and $\gamma_i$, and since $x_i(0)\geq 0$, we have $r_i(1) \geq r_i(0) \geq 0$. By \eqref{eq:sir_xi} and Assumption \ref{assm:1}, we have $r_i(1) \leq r_i(0) + x_i(0) \leq 1$.
Adding up \eqref{eq:sir_si}-\eqref{eq:sir_ri}, gives that $s_i(1) + x_i(1) + r_i(1) = s_i(0) + x_i(0) + r_i(0)$, which by assumption equals~1.

Now assume for an arbitrary $k$, $s_i(k), x_i(k), r_i(k) \in [0,1]$ and $s_i(k) + x_i(k) + r_i(k) = 1$. Following the exact same arguments as for the base case except replacing $0$ with $k$ and $1$ with $k+1$, it can be shown that $s_i(k+1), x_i(k+1), r_i(k+1) \in [0,1]$ and $s_i(k+1) + x_i(k+1) + r_i(k+1) = 1$. Therefore, by induction, $s_i(k), x_i(k), r_i(k) \in [0,1]$
and
$s_i(k) + x_i(k) + r_i(k) = 1$ 
for all $k\geq0$ and $i\in [n]$. In addition, if $s_i(0) > 0$, the above argument implies that $s_i(k)>0$ for all $i\in [n]$ and for every finite $k$. As a result, the matrix $A_k$ remains irreducible for all finite $k$.

2) By 1) and the non-negativity of the $\beta_{ij}$'s, it follows that $ h \big[-s_i(k) \sum^n_{j=1} \beta_{ij} x_j(k) \big]\leq 0$ for all $k\geq 0$. Therefore, from \eqref{eq:sir_si}, we have $ s_i(k+1) \leq s_i(k)$. 

3) Since the rate of change of $s(k)$, $- h \diag(s(k)) B  x(k)$, is non-positive for all $k\geq0$ and $s(k) $ is lower bounded by zero, by  1), we conclude that $\lim_{k\rightarrow \infty}s(k)$ exists. Therefore 
\begin{equation}\label{eq:lim_schange}
    \lim_{k\rightarrow \infty}- h \diag(s(k)) B x(k) = \0,
\end{equation} 
where $\0$ is the vector of dimension $n$ with all entries equal to $0$. 
Therefore, $\lim_{k\rightarrow \infty} x(k+1) - x(k) = \lim_{k\rightarrow \infty}-h\diag(\gamma) x(k)$. Thus, by the assumption that $h \gamma_i>0$ for all $i \in [n]$ \rev{and 1), we have that} 
$$\rev{\lim_{k\rightarrow \infty} x_i(k+1) - x_i(k) = \lim_{k\rightarrow \infty}-h \gamma_i x_i(k)\leq 0,}$$
\rev{where equality holds only if $\lim_{k \to \infty}x_i(k)=0$ and if the inequality is strict, $x_i(k)$ is decreasing. Therefore,}  $\lim_{k \to \infty} x_i(k) = 0$, for all $i \in [n]$.

4) Recall that by Assumption \ref{assm:1}, $A_k$ is an irreducible non-negative matrix and thus by the Perron-Frobenius Theorem for irreducible non-negative matrices \cite[Theorem~2.7 and Lemma~2.4]{varga}, $\lambda^{A_k}_{\max} = \rho(A_k)$, where $\rho(\cdot)$ indicates the spectral radius. By \cite[Theorem~2.7]{varga}, $\rho(A_k)$ increases when any entry of $A_k$ increases. Therefore, by 2) and since $A_k$ is defined as $\mathbf{I}_n + h \diag(s^k) B - h \diag(\gamma)$, we have that 
$$\rho(A_k) \geq \rho(A_{k+1}),$$
which implies
$$\lambda^{A_k}_{\max} \geq \lambda^{A_{k+1}}_{\max}.$$

5) There are two possible types of equilibria for the SIR model: i) $\lim_{k\rightarrow \infty}s(k) = \0$, or ii)~$\lim_{k\rightarrow \infty}s(k) = s^* \neq \0$. We explore the two cases separately. 

i) If $\lim_{k\rightarrow \infty}s(k) = \0$, then the rate of change of $x(k)$ \rev{will become} 
$-h\diag(\gamma)x(k)$. Therefore, by 
the definition of \rev{$A_{k}$ and the assumption that $h\gamma_i < 1, \forall i \in [n]$},
there exists a $\bar{k}$ such that $\lambda^{A_{k}}_{\max}< 1$ for all $k \geq \bar{k}$. 

ii) If 
$\lim_{k\rightarrow \infty}s(k) = s^* \neq \0$, then, by 3), for any $(s(0),x(0),r(0))$ the system dynamics converge to some equilibria of the form $(s^*,\0,\1-s^*)$, for some nonzero $s^*$. 
Define 
\begin{equation}\label{eq:eps}
    \epsilon_s(k):= s(k) - s^* \text{ and } \epsilon_x(k):= x(k) - \0.
\end{equation}By 2) and 1), respectively, we know that $\epsilon_s(k)\geq0$ and $\epsilon_x(k)\geq 0$ for all $k\geq 0$. Furthermore, we know that $\epsilon_s(k+1)\leq \epsilon_s(k)$ for all $k\geq 0$, $\lim_{k\rightarrow \infty}\epsilon_s(k) = \0$, and $\lim_{k\rightarrow \infty}\epsilon_x(k) = \0$. 

Let $\lambda^{A^*}_{\max}$ be the maximum eigenvalue 
of $\mathbf{I}_n + h \diag(s^*) B - h \diag(\gamma)$ with corresponding normalized left eigenvector $w^*$, that is, 
\begin{equation}\label{eq:eigen}
     w^{*\top}(\mathbf{I}_n + h \diag(s^*) B - h \diag(\gamma)) = \lambda^{A^*}_{\max} w^{*\top}.
\end{equation}
\rev{Note that since $\mathbf{I}_n + h \diag(s^*) B - h \diag(\gamma)$ is irreducible, $w^*_i>0$, for all $i \in [n]$, by the Perron-Frobenius Theorem.}

\rev{Assume on the contrary that $\lambda^{A^*}_{\max}\geq 1$}. Left multiplying the equation for $x(k+1)$ in \eqref{eq:sir_dt_x} by $w^{*\top}$ and using 
\eqref{eq:eps} and \eqref{eq:eigen}, we obtain
\begin{align*}
    & w^{*\top}\epsilon_x(k+1) 
    \\ & \quad = w^{*\top} [\mathbf{I}_n + h \diag(\epsilon_s(k) + s^*) B - h \diag(\gamma)] \epsilon_x(k) \\
    & \quad =\lambda^{A^*}_{\max}w^{*\top}\epsilon_x(k)  + w^{*\top} h \diag(\epsilon_s(k)) B \epsilon_x(k) \\ 
    & \quad \rev{\geq} w^{*\top}\epsilon_x(k)  + w^{*\top} h \diag(\epsilon_s(k)) B \epsilon_x(k),
\end{align*}
where the last equality holds since 
$\lambda^{A^*}_{\max}\rev{\geq}1$.
Thus, $$w^{*\top}(\epsilon_x(k+1)-\epsilon_x(k)) \rev{\geq} w^{*\top} h \diag(\epsilon_s(k)) B \epsilon_x(k) \geq 0.$$ 
\rev{If either inequality is strict, clearly, $\lim_{k\rightarrow \infty}x(k) \neq \0$, which is a contradiction to 3). 
If equality holds, $w^{*\top}(\epsilon_x(k+1)-\epsilon_x(k))=0$ implies $\epsilon_x(k+1)=\epsilon_x(k)$, since $w^*_i>0$, for all $i \in [n]$. 
Then, since by assumption, $x_i(0) > 0$ for some~$i$, $\lim_{k\rightarrow \infty}x(k) \neq \0$, again a contradiction to 3). }
% ...
Therefore, there exists a $\bar{k}$ such that $\lambda^{A_{k}}_{\max}< 1$ for all $k \geq \bar{k}$.

6) Since, by 5), there exists a $\bar{k}$ such that $\lambda^{A_{k}}_{\max} < 1$ for all $k \geq \bar{k}$, and we know that $\lambda^{A_{k}}_{\max} = \rho(A_{k})\geq 0$ by Assumption~\ref{assm:1}, we have 
\begin{equation}\label{eq:conv}
    \lim_{k \to \infty} \frac{\|x(k+1)\|}{\|x(k)\|} = \frac{\|A_k x(k)\|}{\|x(k)\|} \leq \lambda^{A_{k}}_{\max}<1.
\end{equation}
Therefore, for $k \geq \bar{k}$, $x(k)$ converges linearly to $\0$.
\end{proof}

\sloppy
\bibliographystyle{plainnat}
\bibliography{refs,refs_new}

\begin{thebibliography}{52}
\providecommand{\natexlab}[1]{#1}
\providecommand{\url}[1]{\texttt{#1}}
\expandafter\ifx\csname urlstyle\endcsname\relax
  \providecommand{\doi}[1]{doi: #1}\else
  \providecommand{\doi}{doi: \begingroup \urlstyle{rm}\Url}\fi

\bibitem[Ansumali et~al.(2020)Ansumali, Kaushal, Kumar, Prakash, and
  Vidyasagar]{ansumali2020modelling}
Santosh Ansumali, Shaurya Kaushal, Aloke Kumar, Meher~K Prakash, and
  M~Vidyasagar.
\newblock Modelling a pandemic with asymptomatic patients, impact of lockdown
  and herd immunity, with applications to {SARS}-{C}o{V}-2.
\newblock \emph{Annual Reviews in Control}, 2020.

\bibitem[Boyd et~al.(2007)Boyd, Kim, Vandenberghe, and
  Hassibi]{boyd2007tutorial}
Stephen Boyd, Seung-Jean Kim, Lieven Vandenberghe, and Arash Hassibi.
\newblock A tutorial on geometric programming.
\newblock \emph{Optimization and {E}ngineering}, 8\penalty0 (1):\penalty0 67,
  2007.

\bibitem[Burnell et~al.(2020)Burnell, Damen, and Hoburg]{gpkit}
Edward Burnell, Nicole~B Damen, and Warren Hoburg.
\newblock Gpkit: A human-centered approach to convex optimization in
  engineering design.
\newblock In \emph{Proceedings of the 2020 CHI Conference on Human Factors in
  Computing Systems}, pages 1--13, 2020.

\bibitem[Calafiore et~al.(2020)Calafiore, Novara, and
  Possieri]{calafiore2020time}
Giuseppe~C Calafiore, Carlo Novara, and Corrado Possieri.
\newblock A time-varying {SIRD} model for the {COVID}-19 contagion in {Italy}.
\newblock \emph{Annual Reviews in Control}, 2020.

\bibitem[Carli et~al.(2020)Carli, Cavone, Epicoco, Scarabaggio, and
  Dotoli]{carli2020model}
Raffaele Carli, Graziana Cavone, Nicola Epicoco, Paolo Scarabaggio, and
  Mariagrazia Dotoli.
\newblock Model predictive control to mitigate the {COVID}-19 outbreak in a
  multi-region scenario.
\newblock \emph{Annual Reviews in Control}, 2020.

\bibitem[Casella(2020)]{casella2020can}
Francesco Casella.
\newblock Can the {COVID-19} epidemic be managed on the basis of daily data?,
  2020.
\newblock arXiv preprint arXiv:2003.06967.

\bibitem[Chen and Qiu(2020)]{chen2020scenario}
Xiaohui Chen and Ziyi Qiu.
\newblock Scenario analysis of non-pharmaceutical interventions on global
  {COVID-19} transmissions, 2020.
\newblock arXiv preprint arXiv:2004.04529.

\bibitem[Cohen and Kupferschmidt(2020)]{cohen2020countries}
Jon Cohen and Kai Kupferschmidt.
\newblock Countries test tactics in war against {COVID-19}, 2020.
\newblock {A}merican Association for the Advancement of Science.

\bibitem[Dabarov et~al.(2020)Dabarov, Sharipov, Dadlani, Kumar, Saad, and
  Hong]{dabarov2020heterogeneous}
Aldiyar Dabarov, Madiyar Sharipov, Aresh Dadlani, Muthukrishnan~Senthil Kumar,
  Walid Saad, and Choong~Seon Hong.
\newblock Heterogeneous projection of disruptive malware prevalence in mobile
  social networks.
\newblock \emph{IEEE Communications Letters}, 24\penalty0 (8):\penalty0
  1673--1677, 2020.

\bibitem[Dadlani et~al.(2017)Dadlani, Kumar, Maddi, and Kim]{dadlani2017mean}
Aresh Dadlani, Muthukrishnan~Senthil Kumar, Manikanta~Gowtham Maddi, and Kiseon
  Kim.
\newblock Mean-field dynamics of inter-switching memes competing over multiplex
  social networks.
\newblock \emph{IEEE Communications Letters}, 21\penalty0 (5):\penalty0
  967--970, 2017.

\bibitem[Della~Rossa et~al.(2020)Della~Rossa, Salzano, Di~Meglio, De~Lellis,
  Coraggio, Calabrese, Guarino, Cardona-Rivera, De~Lellis, Liuzza,
  et~al.]{della2020network}
Fabio Della~Rossa, Davide Salzano, Anna Di~Meglio, Francesco De~Lellis, Marco
  Coraggio, Carmela Calabrese, Agostino Guarino, Ricardo Cardona-Rivera, Pietro
  De~Lellis, Davide Liuzza, et~al.
\newblock A network model of {Italy} shows that intermittent regional
  strategies can alleviate the {COVID}-19 epidemic.
\newblock \emph{Nature Communications}, 11\penalty0 (1):\penalty0 1--9, 2020.

\bibitem[Draief and Massouli(2010)]{draief2010epidemics}
Moez Draief and Laurent Massouli.
\newblock \emph{Epidemics and rumours in complex networks}.
\newblock Cambridge University Press, 2010.

\bibitem[Enyioha et~al.(2015)Enyioha, Jadbabaie, Preciado, and
  Pappas]{enyioha2015distributed}
Chinwendu Enyioha, Ali Jadbabaie, Victor Preciado, and George Pappas.
\newblock Distributed resource allocation for control of spreading processes.
\newblock In \emph{Proceedings of the 2015 European Control Conference (ECC)},
  pages 2216--2221. IEEE, 2015.

\bibitem[Eshghi et~al.(2014)Eshghi, Khouzani, Sarkar, and
  Venkatesh]{eshghi2014optimal}
Soheil Eshghi, MHR Khouzani, Saswati Sarkar, and Santosh~S Venkatesh.
\newblock Optimal patching in clustered malware epidemics.
\newblock \emph{IEEE/ACM Transactions on Networking}, 24\penalty0 (1):\penalty0
  283--298, 2014.

\bibitem[Giordano et~al.(2020)Giordano, Blanchini, Bruno, Colaneri, Di~Filippo,
  Di~Matteo, and Colaneri]{giordano2020modelling}
Giulia Giordano, Franco Blanchini, Raffaele Bruno, Patrizio Colaneri,
  Alessandro Di~Filippo, Angela Di~Matteo, and Marta Colaneri.
\newblock Modelling the {COVID-19} epidemic and implementation of
  population-wide interventions in {I}taly.
\newblock \emph{Nature Medicine}, pages 1--6, 2020.

\bibitem[Grundel et~al.(2020)Grundel, Heyder, Hotz, Ritschel, Sauerteig, and
  Worthmann]{grundel2020much}
Sara Grundel, Stefan Heyder, Thomas Hotz, Tobias~KS Ritschel, Philipp
  Sauerteig, and Karl Worthmann.
\newblock How much testing and social distancing is required to control
  {COVID}-19? {S}ome insight based on an age-differentiated compartmental
  model.
\newblock \emph{arXiv preprint arXiv:2011.01282}, 2020.

\bibitem[Han et~al.(2015)Han, Preciado, Nowzari, and Pappas]{han2015data}
Shuo Han, Victor~M Preciado, Cameron Nowzari, and George~J Pappas.
\newblock Data-driven network resource allocation for controlling spreading
  processes.
\newblock \emph{IEEE Transactions on Network Science and Engineering},
  2\penalty0 (4):\penalty0 127--138, 2015.

\bibitem[Hethcote(2000)]{hethcote2000mathematics}
Herbert~W Hethcote.
\newblock The mathematics of infectious diseases.
\newblock \emph{SIAM Review}, 42\penalty0 (4):\penalty0 599--653, 2000.

\bibitem[Hota and Sundaram(2019)]{hota2019game}
Ashish~R Hota and Shreyas Sundaram.
\newblock Game-theoretic vaccination against networked {SIS} epidemics and
  impacts of human decision-making.
\newblock \emph{IEEE Transactions on Control of Network Systems}, 6\penalty0
  (4):\penalty0 1461--1472, 2019.

\bibitem[Hota et~al.(2020)Hota, Sneh, and Gupta]{hota2020impacts}
Ashish~R Hota, Tanya Sneh, and Kavish Gupta.
\newblock Impacts of game-theoretic activation on epidemic spread over
  dynamical networks.
\newblock \emph{arXiv preprint arXiv:2011.00445}, 2020.

\bibitem[Hu et~al.(2020)Hu, Song, Xu, Jin, Chen, Xu, Ma, Chen, Lin, Zheng,
  et~al.]{hu2020clinical}
Zhiliang Hu, Ci~Song, Chuanjun Xu, Guangfu Jin, Yaling Chen, Xin Xu, Hongxia
  Ma, Wei Chen, Yuan Lin, Yishan Zheng, et~al.
\newblock Clinical characteristics of 24 asymptomatic infections with
  {COVID-19} screened among close contacts in {N}anjing, {C}hina.
\newblock \emph{Science China Life Sciences}, 63\penalty0 (5):\penalty0
  706--711, 2020.

\bibitem[K{\"o}hler et~al.(2020)K{\"o}hler, Schwenkel, Koch, Berberich, Pauli,
  and Allg{\"o}wer]{kohler2020robust}
Johannes K{\"o}hler, Lukas Schwenkel, Anne Koch, Julian Berberich, Patricia
  Pauli, and Frank Allg{\"o}wer.
\newblock Robust and optimal predictive control of the {COVID}-19 outbreak.
\newblock \emph{Annual Reviews in Control}, 2020.

\bibitem[Liu et~al.(2019)Liu, Par{\'e}, Nedi{\'c}, Tang, Beck, and
  Ba{\c{s}}ar]{liu2019analysis}
Ji~Liu, Philip~E Par{\'e}, Angelia Nedi{\'c}, Choon~Yik Tang, Carolyn~L Beck,
  and Tamer Ba{\c{s}}ar.
\newblock Analysis and control of a continuous-time bi-virus model.
\newblock \emph{IEEE Transactions on Automatic Control}, 64\penalty0
  (12):\penalty0 4891--4906, 2019.

\bibitem[Lofberg(2004)]{lofberg2004yalmip}
Johan Lofberg.
\newblock Yalmip: A toolbox for modeling and optimization in matlab.
\newblock In \emph{2004 IEEE international conference on robotics and
  automation (IEEE Cat. No. 04CH37508)}, pages 284--289. IEEE, 2004.

\bibitem[Mai et~al.(2018)Mai, Battou, and Mills]{mai2018distributed}
Van~Sy Mai, Abdella Battou, and Kevin Mills.
\newblock Distributed algorithm for suppressing epidemic spread in networks.
\newblock \emph{IEEE Control Systems Letters}, 2\penalty0 (3):\penalty0
  555--560, 2018.

\bibitem[Massonis et~al.(2020)Massonis, Banga, and
  Villaverde]{massonis2020structural}
Gemma Massonis, Julio~R Banga, and Alejandro~F Villaverde.
\newblock Structural identifiability and observability of compartmental models
  of the {COVID}-19 pandemic.
\newblock \emph{Annual Reviews in Control}, 2020.

\bibitem[Max et~al.(2020)Max, Hannah, Ortiz-Ospina, and Hasell]{Data_COVID}
Roser Max, Ritchie Hannah, Esteban Ortiz-Ospina, and Joe Hasell.
\newblock Coronavirus pandemic ({COVID}-19), 2020.
\newblock URL \url{https://ourworldindata.org/coronavirus-testing}.
\newblock {U}niversity of Oxford, Accessed: 2020-05-18.

\bibitem[Mei et~al.(2017)Mei, Mohagheghi, Zampieri, and Bullo]{mei2017dynamics}
Wenjun Mei, Shadi Mohagheghi, Sandro Zampieri, and Francesco Bullo.
\newblock On the dynamics of deterministic epidemic propagation over networks.
\newblock \emph{Annual Reviews in Control}, 44:\penalty0 116--128, 2017.

\bibitem[Miller(2012)]{miller2012note}
Joel~C Miller.
\newblock A note on the derivation of epidemic final sizes.
\newblock \emph{Bulletin of {M}athematical {B}iology}, 74\penalty0
  (9):\penalty0 2125--2141, 2012.

\bibitem[Morato et~al.(2020)Morato, Bastos, Cajueiro, and
  Normey-Rico]{morato2020optimal}
Marcelo~M Morato, Saulo~B Bastos, Daniel~O Cajueiro, and Julio~E Normey-Rico.
\newblock An optimal predictive control strategy for {COVID}-19
  ({SARS}-{C}o{V}-2) social distancing policies in {Brazil}.
\newblock \emph{Annual Reviews in Control}, 50:\penalty0 417--431, 2020.

\bibitem[Nowzari et~al.(2016)Nowzari, Preciado, and
  Pappas]{nowzari2016analysis}
Cameron Nowzari, Victor~M Preciado, and George~J Pappas.
\newblock Analysis and control of epidemics: A survey of spreading processes on
  complex networks.
\newblock \emph{IEEE Control Systems}, 36\penalty0 (1):\penalty0 26--46, 2016.

\bibitem[Ogura and Preciado(2016)]{ogura2016efficient}
Masaki Ogura and Victor~M Preciado.
\newblock Efficient containment of exact {SIR} {Markovian} processes on
  networks.
\newblock In \emph{Proceedings of the 2016 IEEE 55th Conference on Decision and
  Control (CDC)}, pages 967--972. IEEE, 2016.

\bibitem[Osthus et~al.(2017)Osthus, Hickmann, Caragea, Higdon, and
  Del~Valle]{osthus2017forecasting}
Dave Osthus, Kyle~S Hickmann, Petru{\c{t}}a~C Caragea, Dave Higdon, and Sara~Y
  Del~Valle.
\newblock Forecasting seasonal influenza with a state-space {SIR} model.
\newblock \emph{The Annals of Applied Statistics}, 11\penalty0 (1):\penalty0
  202, 2017.

\bibitem[Ota(2020)]{ota2020will}
Miyo Ota.
\newblock Will we see protection or reinfection in {COVID-19}?
\newblock \emph{Nature Reviews Immunology}, 20\penalty0 (6):\penalty0 351--351,
  2020.

\bibitem[Par{\'e} et~al.(2020{\natexlab{a}})Par{\'e}, Beck, and
  Ba{\c{s}}ar]{pare2020modeling}
Philip~E Par{\'e}, Carolyn~L Beck, and Tamer Ba{\c{s}}ar.
\newblock Modeling, estimation, and analysis of epidemics over networks: {A}n
  overview.
\newblock \emph{Annual Reviews in Control}, 2020{\natexlab{a}}.

\bibitem[Par{\'e} et~al.(2020{\natexlab{b}})Par{\'e}, Liu, Beck, Kirwan, and
  Ba{\c{s}}ar]{pare2018analysis}
Philip~E Par{\'e}, Ji~Liu, Carolyn~L Beck, Barrett~E Kirwan, and Tamer
  Ba{\c{s}}ar.
\newblock Analysis, estimation, and validation of discrete-time epidemic
  processes.
\newblock \emph{IEEE Transactions on Control Systems Technology}, 28\penalty0
  (1):\penalty0 79--93, 2020{\natexlab{b}}.

\bibitem[Pastor-Satorras et~al.(2015)Pastor-Satorras, Castellano, Van~Mieghem,
  and Vespignani]{pastor2015epidemic}
Romualdo Pastor-Satorras, Claudio Castellano, Piet Van~Mieghem, and Alessandro
  Vespignani.
\newblock Epidemic processes in complex networks.
\newblock \emph{Reviews of Modern Physics}, 87\penalty0 (3):\penalty0 925,
  2015.

\bibitem[Prasse and Van~Mieghem(2020)]{prasse2020network}
Bastian Prasse and Piet Van~Mieghem.
\newblock Network reconstruction and prediction of epidemic outbreaks for
  general group-based compartmental epidemic models.
\newblock \emph{IEEE Transactions on Network Science and Engineering}, 2020.

\bibitem[Preciado et~al.(2014)Preciado, Zargham, Enyioha, Jadbabaie, and
  Pappas]{preciado2014optimal}
Victor~M Preciado, Michael Zargham, Chinwendu Enyioha, Ali Jadbabaie, and
  George~J Pappas.
\newblock Optimal resource allocation for network protection against spreading
  processes.
\newblock \emph{IEEE Transactions on Control of Network Systems}, 1\penalty0
  (1):\penalty0 99--108, 2014.

\bibitem[Ram{\'\i}rez-Llanos and Mart{\'\i}nez(2018)]{ramirez2018distributed}
Eduardo Ram{\'\i}rez-Llanos and Sonia Mart{\'\i}nez.
\newblock Distributed discrete-time optimization algorithms with applications
  to resource allocation in epidemics control.
\newblock \emph{Optimal Control Applications and Methods}, 39\penalty0
  (1):\penalty0 160--180, 2018.

\bibitem[Sahneh and Scoglio(2014)]{sahneh2014competitive}
Faryad~Darabi Sahneh and Caterina Scoglio.
\newblock Competitive epidemic spreading over arbitrary multilayer networks.
\newblock \emph{Physical Review E}, 89\penalty0 (6):\penalty0 062817, 2014.

\bibitem[Sahneh et~al.(2017)Sahneh, Vajdi, Melander, and
  Scoglio]{sahneh2017contact}
Faryad~Darabi Sahneh, Aram Vajdi, Joshua Melander, and Caterina~M Scoglio.
\newblock Contact adaption during epidemics: A multilayer network formulation
  approach.
\newblock \emph{IEEE Transactions on Network Science and Engineering},
  6\penalty0 (1):\penalty0 16--30, 2017.

\bibitem[Salehi et~al.(2015)Salehi, Sharma, Marzolla, Magnani, Siyari, and
  Montesi]{salehi2015spreading}
Mostafa Salehi, Rajesh Sharma, Moreno Marzolla, Matteo Magnani, Payam Siyari,
  and Danilo Montesi.
\newblock Spreading processes in multilayer networks.
\newblock \emph{IEEE Transactions on Network Science and Engineering},
  2\penalty0 (2):\penalty0 65--83, 2015.

\bibitem[Song et~al.(2020)Song, Wang, Zhou, He, Zhu, Wang, Tang, and
  Eisenberg]{song2020epidemiological}
Peter~X Song, Lili Wang, Yiwang Zhou, Jie He, Bin Zhu, Fei Wang, Lu~Tang, and
  Marisa Eisenberg.
\newblock An epidemiological forecast model and software assessing
  interventions on {COVID-19} epidemic in {C}hina, 2020.
\newblock medRxiv.

\bibitem[Varga(2000)]{varga}
Richard~S. Varga.
\newblock \emph{Matrix Iterative Analysis}.
\newblock Springer-Verlag, 2000.

\bibitem[Vrabac et~al.(2020)Vrabac, Par\'{e}, Sandberg, and Johansson]{ciss}
Damir Vrabac, Philip~E Par\'{e}, Henrik Sandberg, and Karl~H Johansson.
\newblock Overcoming challenges for estimating virus spread dynamics from data.
\newblock In \emph{Proceedings of the 54th Annual Conference on Information
  Sciences and Systems (CISS)}, pages 1--6, 2020.

\bibitem[Wan et~al.(2007)Wan, Roy, and Saberi]{wan2007network}
Yan Wan, Sandip Roy, and Ali Saberi.
\newblock Network design problems for controlling virus spread.
\newblock In \emph{Proceedings of the 46th IEEE Conference on Decision and
  Control}, pages 3925--3932. IEEE, 2007.

\bibitem[Wan et~al.(2008)Wan, Roy, and Saberi]{wan2008designing}
Yan Wan, Sandip Roy, and Ali Saberi.
\newblock Designing spatially heterogeneous strategies for control of virus
  spread.
\newblock \emph{IET Systems Biology}, 2\penalty0 (4):\penalty0 184--201, 2008.

\bibitem[Yang et~al.(2017)Yang, Yang, and Tang]{yang2017bi}
Lu-Xing Yang, Xiaofan Yang, and Yuan~Yan Tang.
\newblock A bi-virus competing spreading model with generic infection rates.
\newblock \emph{IEEE Transactions on Network Science and Engineering},
  5\penalty0 (1):\penalty0 2--13, 2017.

\bibitem[Zaman et~al.(2008)Zaman, Kang, and Jung]{zaman2008stability}
Gul Zaman, Yong~Han Kang, and Il~Hyo Jung.
\newblock Stability analysis and optimal vaccination of an {SIR} epidemic
  model.
\newblock \emph{BioSystems}, 93\penalty0 (3):\penalty0 240--249, 2008.

\bibitem[Zhang et~al.(2021)Zhang, Leung, Butler, Par{\'e},
  et~al.]{zhang2021estimation}
Ciyuan Zhang, Humphrey Leung, Brooks Butler, Philip Par{\'e}, et~al.
\newblock Estimation and distributed eradication of {SIR} epidemics on
  networks.
\newblock \emph{arXiv preprint arXiv:2102.12549}, 2021.

\bibitem[Zhuang and Ya{\u{g}}an(2016)]{zhuang2016information}
Yong Zhuang and Osman Ya{\u{g}}an.
\newblock Information propagation in clustered multilayer networks.
\newblock \emph{IEEE Transactions on Network Science and Engineering},
  3\penalty0 (4):\penalty0 211--224, 2016.

\end{thebibliography}

\clearpage
\onecolumn

\section*{Supplementary Appendix: Additional Results from the Analysis of Real Testing Data}\label{section:supplement}

\subsection{Learned Parameters}

The following tables show the learned infection rate ($\beta_{ij}$) parameters computed by solving \eqref{eq:ls_opt_new} for different sub-intervals for $\alpha = 40$ and testing data from March 1st 2020 (day 1) to 29th May 2020 (day 90). 

\begin{table*}[htb]
    \begin{minipage}{.48\linewidth}
      \centering
        \begin{tabular}{c|c|c|c|c|c}
        $-$ & DE  & FR & AT & IT & CH  \\ \hline
        DE & 0.0724  &  0  &  0  & 0  &  0.0158 \\
        FR & 0 &  0.015  &  0  & 0.035  &  0 \\
        AT & 0  &  0  &  0  & 0.0498  &  0 \\
        IT & 0  &  0.1  &  0  &  0.0377  &  0 \\
        CH & 0  &  0.020 &  0  &   0.009  &  0 \\ \hline
    \end{tabular}
    \caption{Learned $\beta_{ij}$ parameters for Days 1-15.}
    \label{tab:contact_12}
    \end{minipage}
    ~
    \begin{minipage}{.48\linewidth}
      \centering
        \begin{tabular}{c|c|c|c|c|c}
        $-$ & DE  & FR & AT & IT & CH  \\ \hline
        DE & 0.11  &  0.0007  &  0.004  & 0  &  0 \\
        FR & 0.18 &  0.047  &  0  & 0  &  0 \\
        AT & 0  &  0  &  0.027  & 0  &  0 \\
        IT & 0  &  0.095  &  0  &  0 &  0 \\
        CH & 0.168  &  0 &  0 &  0.021 & 0 \\ \hline
    \end{tabular}
    \caption{Learned $\beta_{ij}$ parameters for Days 31-60.}
    \label{tab:contact_78}
        \end{minipage}
\end{table*}

\begin{table*}[htb]
    \begin{minipage}{.48\linewidth}
      \centering
        \begin{tabular}{c|c|c|c|c|c}
        $-$ & DE  & FR & AT & IT & CH  \\ \hline
        DE & 0.018  &  0.001  &  0.021  & 0  &  0.003 \\
        FR & 0 &  0  &  0  & 0  &  0.099 \\
        AT & 0  &  0  &  0.017  & 0  &  0 \\
        IT & 0  &  0  &  0.047  &  0  &  0 \\
        CH & 0  &  0 &  0.063  &   0  &  0.003 \\ \hline
    \end{tabular}
    \caption{Learned $\beta_{ij}$ parameters for Days 61-90.}
    \label{tab:contact_12}
    \end{minipage}
    ~
    \begin{minipage}{.48\linewidth}
      \centering
        \begin{tabular}{c|c|c|c|c|c}
        $-$ & DE  & FR & AT & IT & CH  \\ \hline
        DE & 0  &  0.001  &  0.039  & 0  &  0.004 \\
        FR & 0 &  0  &  0  & 0  &  0.057 \\
        AT & 0  &  0  &  0  & 0.003  &  0 \\
        IT & 0  &  0  &  0  &  0.002 &  0.034 \\
        CH & 0  &  0 &  0 &  0 & 0.026 \\ \hline
    \end{tabular}
    \caption{Learned $\beta_{ij}$ parameters for Days 16-30.}
    \label{tab:contact_78}
        \end{minipage}
\end{table*}

\subsection{Inference with different values of $\alpha$}

\begin{figure*}[htb]
	\centering
	\includegraphics[scale=0.7]{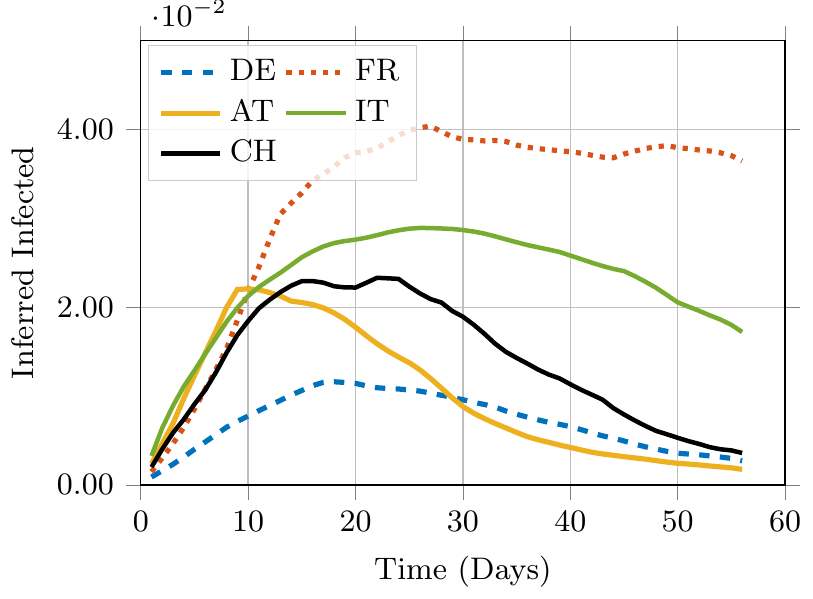}~~~
	\includegraphics[scale=0.7]{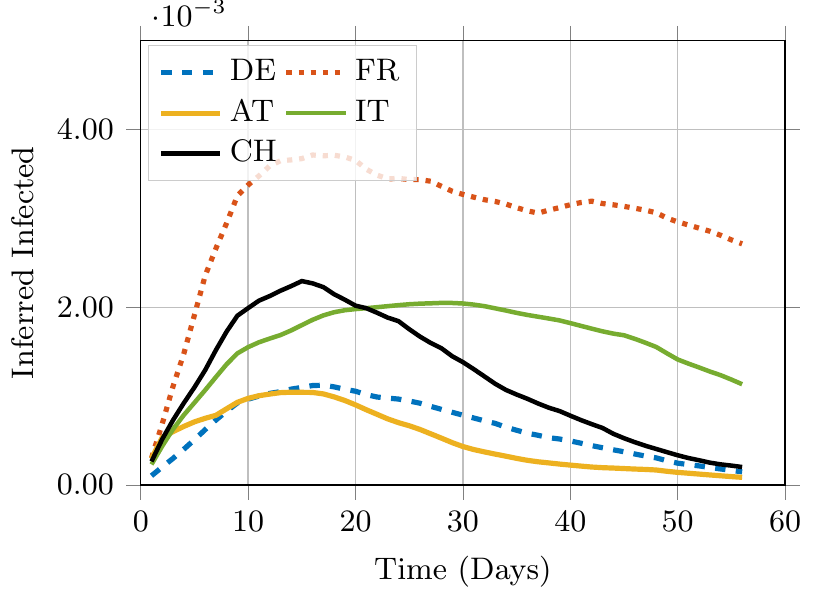}~
	\caption{Inferred states with $\alpha = 100$ and $\alpha = 1000$, respectively.}
    \label{fig:supp_alpha}
\end{figure*}

The above figure shows the inferred infected proportions from March 16th 2020 to May 12th 2020 assuming that the initial infected proportions are zero. Larger value of $\alpha$ indicates that the infected proportion is smaller in magnitude. 

\subsection{Inference with delay factor}

\begin{figure*}[htb]
	\centering
	\includegraphics[scale=0.7]{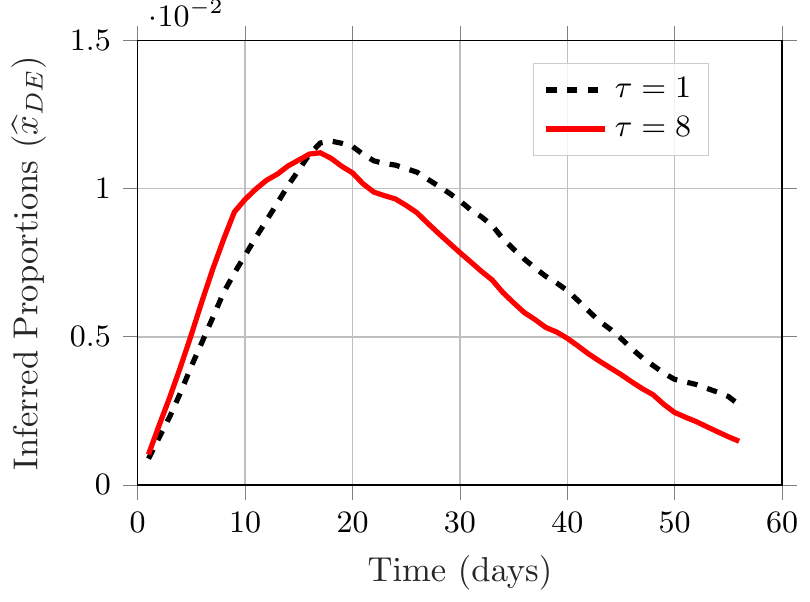}~
	\includegraphics[scale=0.7]{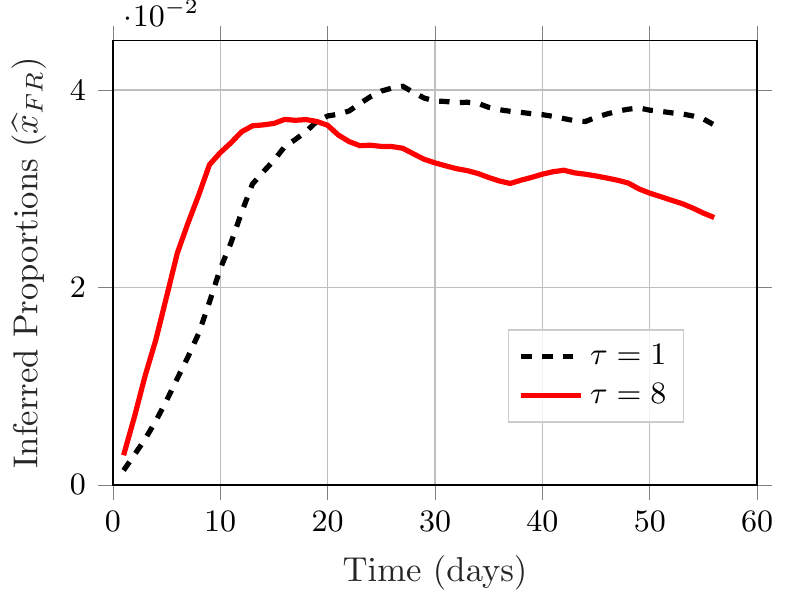}~
	\includegraphics[scale=0.7]{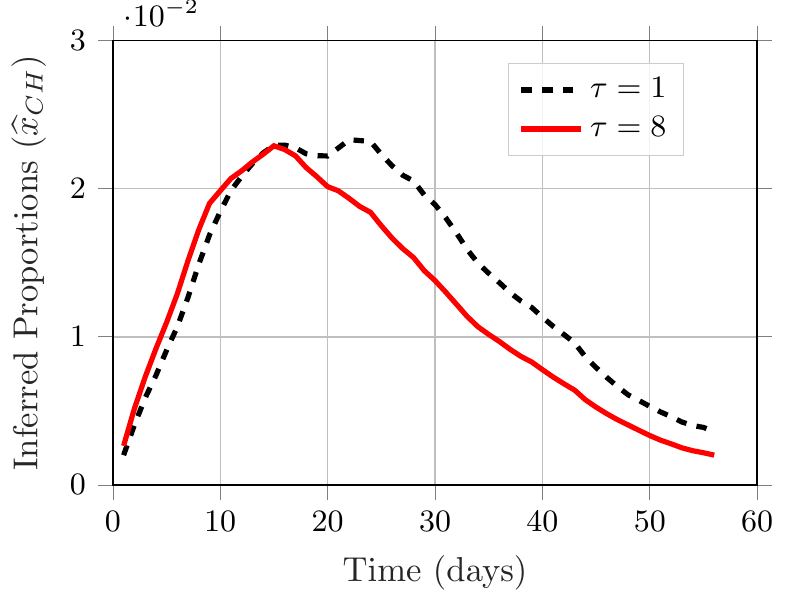}
	\caption{Inferred states with and without delay and with initial infected proportion $0$ and $\alpha = 100$.}
    \label{fig:supp_tau}
\end{figure*}

The above figure shows that the delay factor primarily shifts the inferred states to the left, indicating that the infected proportion was larger than anticipated in the early stage of the pandemic.

\end{document}